\documentclass[11pt]{article}
 
\usepackage[utf8]{inputenc}
\usepackage{authblk}
\usepackage{amsthm}
\usepackage{amsmath}
\usepackage{amssymb}
\usepackage{amsfonts}
\usepackage{mathtools}
\usepackage{fullpage}

\usepackage{tabularx}

\usepackage{thm-restate}

\usepackage{enumitem}
\usepackage[pagebackref]{hyperref}
\usepackage[sort,numbers]{natbib}

\usepackage{tikz}
\newcommand{\ETH}{$\mathsf{ETH}$ }
\newcommand{\NP}{$\mathsf{NP}$ }

\newcommand{\F}{{\textbf{F}}}
\newcommand{\tw}{{\rm{tw}}}
\newcommand{\pw}{{\rm{pw}}}
\newcommand{\K}{{\textbf{P}}}

\newcommand{\FPT}{$\mathsf{FPT}$ }

\usepackage[capitalize]{cleveref}

\usepackage{todonotes}

\newtheorem{theorem}{Theorem}
\newtheorem{lemma}[theorem]{Lemma}
\newtheorem{observation}[theorem]{Observation}
\newtheorem{remark}[theorem]{Remark}
\newtheorem{conjecture}[theorem]{Conjecture}

\newtheorem{corollary}[theorem]{Corollary}
\newtheorem{proposition}[theorem]{Proposition}

\theoremstyle{definition}
\newtheorem{definition}[theorem]{Definition}

\begin{document}

\title{$\mathcal{P}$-matchings Parameterized by Treewidth}

\author{Juhi~Chaudhary}

\author{Meirav~Zehavi}

\affil{\small Department of Computer Science, Ben-Gurion~University~of~the~Negev, 
Beer-Sheva, 
Israel\\ \texttt{juhic@post.bgu.ac.il,  meiravze@bgu.ac.il}}

\date{}
\maketitle
\begin{abstract}
  A \emph{matching} is a subset of edges in a graph $G$ that do not share an endpoint. A matching $M$ is a \emph{$\mathcal{P}$-matching} if the subgraph of $G$ induced by the endpoints of the edges of $M$ satisfies property $\mathcal{P}$. For example, if the property $\mathcal{P}$ is that of being a matching, being acyclic, or being disconnected, then we obtain an \emph{induced matching}, an \emph{acyclic matching}, and a \emph{disconnected matching}, respectively. In this paper, we analyze the problems of the computation of these matchings from the viewpoint of Parameterized Complexity with respect to the parameter \emph{treewidth}.

\end{abstract}

 Hypothesis.

\section{Introduction}
Matching in graphs is a central topic of Graph Theory and Combinatorial Optimization \cite{lovaz}. Matchings possess both theoretical significance and practical applications, such as the assignment of new physicians to hospitals, students to high schools, clients to server clusters, kidney donors to recipients \cite{manlove}, and so on. Additionally, the field of competitive optimization games on graphs has witnessed substantial growth in recent years, where matching serves as a valuable tool for determining optimal solutions or bounds in such games \cite{bachstein,goddard1}. The study of matchings is closely related to the concept of \emph{edge colorings} as well \cite{baste2020approximating,basteurm,vizing}, and the minimum number of matchings into which the edge set of a graph $G$ can be partitioned is known as the \emph{chromatic index} of $G$ \cite{vizing}. 

Given a graph $G$, \textsc{Maximum Matching} is the problem of finding a matching of maximum size (number of edges) in $G$. A matching $M$ is said to be a \emph{$\mathcal{P}$-matching} if $G[V_{M}]$ (the subgraph of $G$ induced by the endpoints of edges in $M$) has property $\mathcal{P}$, where $\mathcal{P}$ is some graph property. The problem of deciding whether a graph admits a $\mathcal{P}$-matching of a given size has been investigated for
many different properties \cite{basteurm1,francis, goddard, golumbic, panda1, panda, stockmeyer}. If the property $\mathcal{P}$ is that of being a graph, a disjoint union of $K_{2}'s$, a forest, a connected graph, a disconnected graph, or having a unique perfect matching, then a $\mathcal{P}$-matching is a \emph{matching} \cite{micali}, an \emph{induced matching} \cite{stockmeyer}, an \emph{acyclic matching} \cite{goddard}, a \emph{connected matching}\footnote{This name is also used for a different problem where we are asked to find a matching $M$ such that every pair of edges in $M$ has a common edge \cite{cameron1}.}\cite{goddard}, a \emph{disconnected matching}\footnote{In this paper, we are using a different (more general) definition for disconnected matching than the one mentioned in \cite{goddard}.} \cite{goddard}, and a \emph{uniquely restricted matching} \cite{golumbic}, respectively. Notably, only the optimization problem corresponding to matching \cite{micali} and connected matching \cite{goddard} are polynomial-time solvable for a general graph, while the decision problems corresponding to other above-mentioned variants of matching are $\mathsf{NP}$-$\mathsf{complete}$ \cite{goddard,golumbic,stockmeyer}. 

Given a graph $G$ and a positive integer $\ell$, the \textsc{Induced Matching} problem asks whether $G$ has an induced matching of size at least $\ell$. The concept of induced matching was introduced by Stockmeyer and Vazirani as the “risk-free” marriage problem in 1982 \cite{stockmeyer}. Since then, this concept, and the corresponding \textsc{Induced Matching} problem, have been studied extensively due to their wide range of applications and connections to other graph problems \cite{cameron, cooley, klemz, Ko, moser, stockmeyer}. Similarly, the \textsc{Acyclic Matching} problem considers a graph $G$ and a positive integer $\ell$, and asks whether $G$ contains an acyclic matching of size at least $\ell$. Goddard et al. \cite{goddard} introduced the concept of acyclic matching, and since then, it has gained significant popularity in the literature \cite{baste2020approximating,baste2018degenerate,furst2019some,panda1,panda}. For a fixed $c\in \mathbb{N}$, a matching $M$ is \emph{$c$-disconnected} if $G[V_{M}]$ has at least $c$ connected components. In the \textsc{$c$-Disconnected Matching} problem, given a graph $G$ and a positive integer $\ell$, we seek to determine if $G$ contains a $c$-disconnected matching of size at least $\ell$. In the \textsc{$c$-Disconnected Matching} problem, if $c$ is a part of the input, then the problem that arises is known as the \textsc{Disconnected Matching} problem. Goddard et al. \cite{goddard} introduced the concept of disconnected matching along with several other variations of matching and asked about the complexity of determining the maximum size of a matching whose vertex set induces a disconnected graph, which is a restricted version of $c$-disconnected matching studied in this paper. 

Similar to the chromatic index, there is a corresponding notion of edge coloring for other variants of matching also. For example, the \emph{strong chromatic index} is the minimum number of induced matchings (also known as \emph{strong matchings}) into which the edge set of $G$ can be partitioned \cite{faudree}. The \emph{uniquely restricted chromatic index} \cite{basteurm} and the \emph{acyclic chromatic index} \cite{baste2018degenerate} are defined similarly in the literature. While Vizing’s famous theorem \cite{vizing} states that the chromatic index of a simple graph $G$ is either $\rm{\Delta(G)}$ or $\rm{\Delta(G)}+1$, where $\rm{\Delta(G)}$ denotes the maximum degree of a vertex in $G$, two famous open conjectures due to Alon, Sudakov, and Zaks \cite{alon}, and due to Erd\"os and Ne\v set\v ril \cite{erdos} concern upper bounds on the acyclic chromatic index and strong chromatic index in terms of $\rm{\Delta(G)}$, respectively.

The parameter considered in this paper is \emph{treewidth}, a structural parameter that indicates how much a graph resembles a tree. Robertson and Seymour introduced the notion of treewidth in their celebrated work on graph minors \cite{robertson}, and since then, over 4260 papers on google scholar consider treewidth as a parameter in the context of Parameterized Complexity. In practice also, graphs of bounded treewidth appear in many different contexts; for example, many probabilistic networks
appear to have small treewidth \cite{bod}. Thus, concerning the problems studied in this paper, after the \emph{solution size}, treewidth is one of the most natural parameters. In fact, many of the problems investigated in this paper have already been analyzed with respect to treewidth as the parameter.

Formally, the decision versions of the problems associated with the $\mathcal{P}$-matchings studied in this paper are defined below:
\bigskip

\noindent\fbox{ \parbox{160mm}{
		\noindent \underline{\textsc{Induced Matching:}}\\
		\textbf{Input:} An undirected graph $ G$ with $|V(G)|=n$ and a positive integer $\ell$.\\
		\textbf{Question:} Does there exist a set $M\subseteq E(G)$ of cardinality at least $\ell$ such that $G[V_{M}]=M$?}}
		\medskip 

		\noindent\fbox{ \parbox{160mm}{
		\noindent \underline{\textsc{Acyclic Matching:}}\\
		\textbf{Input:} An undirected graph $ G$ with $|V(G)|=n$ and a positive integer $\ell$.\\
		\textbf{Question:} Does there exist a set $X\subseteq V(G)$ of cardinality at least $\ell$ such that $G[X]$ is a forest and $G[X]$ contains a perfect matching?}}
		\medskip

		\noindent\fbox{ \parbox{160mm}{
		\noindent \underline{\textsc{$c$-disconnected Matching:}}\\
		\textbf{Input:} An undirected graph $ G$ with $|V(G)|=n$ and a positive integer $\ell$.\\
		\textbf{Question:} Does there exist a set $M\subseteq E(G)$ of cardinality at least $\ell$ such that $G[V_{M}]$ has at least $c$ connected components for some fixed integer $c\geq 1$?}}
		\medskip

		\noindent\fbox{ \parbox{160mm}{
		\noindent \underline{\textsc{Disconnected Matching:}}\\
		\textbf{Input:} An undirected graph $ G$ with $|V(G)|=n$ and two positive integers $\ell$ and $c$.\\
		\textbf{Question:} Does there exist a set $M\subseteq E(G)$ of cardinality at least $\ell$ such that $G[V_{M}]$ has at least $c$ connected components?}}
		\medskip 
		
\section{Related Work}\label{related}	
In what follows, we present a brief survey of algorithmic results concerning the variants of matchings discussed in this paper.

\medskip

\noindent \textbf{Induced Matching.}
The \textsc{Induced Matching} problem exhibits different computational complexities depending on the class of graphs considered. It is known to be $\mathsf{NP}$-$\mathsf{complete}$ for bipartite graphs of maximum degree $4$ \cite{stockmeyer}, $k$-regular graphs for $k\geq 4$ \cite{zito}, and planar graphs of maximum degree $4$ \cite{Ko}. On the positive side, the problem is known to be polynomial-time solvable for many classes of graphs, such as chordal graphs \cite{cameron}, chordal bipartite graphs \cite{cameron2}, trapezoid graphs, interval-dimension graphs, and cocomparability graphs \cite{golumbic1}. Recently, induced matching on random graphs has been studied by Cooley et al. \cite{cooley}. 
\medskip

From the viewpoint of Parameterized Complexity, in \cite{moser}, Moser and Sikdar showed that \textsc{Induced Matching} is fixed-parameter tractable ($\mathsf{FPT}$) when parameterized by treewidth by developing an $\mathcal{O}(4^{\tw}\cdot n)$-time dynamic programming algorithm. In the same paper (\cite{moser}), when the parameter is the size of the matching $\ell$, \textsc{Induced Matching} was shown to be \FPT for line graphs, planar graphs, bounded-degree graphs, and graphs of girth at least $6$ that include graphs like $C_{4}$-free graphs\footnote{Here, $C_{n}$ denotes a cycle on $n$ vertices.}. On the other hand, for the same parameter, that is, $\ell$, the problem is $\mathsf{W[1]}$-\textsf{hard} for bipartite graphs \cite{moser}. Song \cite{song} showed that given a Hamiltonian cycle
in a Hamiltonian bipartite graph, \textsc{Induced Matching} is $\mathsf{W[1]}$-\textsf{hard} with respect to $\ell$ and cannot be solved in time $n^{o(\sqrt{\ell})}$ unless $\mathsf{W[1]}=\mathsf{FPT}$, where $n$ is the number of vertices in the input graph. \textsc{Induced Matching} with respect to \emph{below guarantee} parameterizations have also been studied \cite{koana,moser2009parameterized,xiao}.
\medskip

\noindent \textbf{Acyclic Matching.}
Baste et al. \cite{baste2018degenerate} demonstrated that finding a maximum cardinality 1-degenerate matching in a graph $G$ is
equivalent to finding a maximum acyclic matching in $G$. \textsc{Acyclic Matching} is known to be $\mathsf{NP}$-$\mathsf{complete}$ for perfect elimination bipartite graphs, a subclass
of bipartite graphs \cite{panda}, star-convex bipartite graphs \cite{panda1}, and dually chordal graphs \cite{panda1}. On the positive side, \textsc{Acyclic Matching} is polynomial-time solvable for chordal graphs \cite{baste2018degenerate} and bipartite permutation graphs \cite{panda}. F\"urst and Rautenbach showed that it is $\mathsf{NP}$-\textsf{hard} to decide whether
a given bipartite graph of maximum degree at most $4$ has a maximum matching
that is acyclic \cite{furst2019some}. In the same paper (\cite{furst2019some}), the authors further characterized the graphs for which every maximum matching is acyclic and gave linear-time algorithms to compute a maximum acyclic matching in graph classes like $P_{4}$-free graphs and $2P_{3}$-free graphs. Additionaly, Panda and Chaudhary \cite{panda1} showed that \textsc{Acyclic Matching} is hard to approximate within factor $n^{1-\epsilon}$ for every $\epsilon>0$ unless $\mathsf{P=NP}$.
\medskip

From the viewpoint of Parameterized Complexity, Hajebi and Javadi \cite{hajebi} discussed the first parameterization results for the \textsc{Acyclic Matching} problem. They showed that \textsc{Acyclic Matching} is $\mathsf{FPT}$ when parameterized by treewidth using Courcelle's theorem. Furthermore, they showed that the problem is $\mathsf{W[1]}$-\textsf{hard} on bipartite graphs when parameterized by the size of the matching. 
However, under the same parameter, the authors showed that the problem is $\mathsf{FPT}$ for line graphs, $C_{4}$-free graphs, and every proper minor-closed class of graphs. In the same paper (\cite{hajebi}), \textsc{Acyclic Matching} was shown to be \FPT when parameterized by the size of the matching plus the number of cycles of length four in the given graph. 

\medskip

\noindent \textbf{$c$-Disconnected Matching and Disconnected Matching.} 
For every fixed integer $c\geq 2$, $c$-\textsc{Disconnected Matching} is known to be $\mathsf{NP}$-$\mathsf{complete}$ even for bounded diameter bipartite graphs \cite{gomes}. On the other hand, for $c=1$, $c$-\textsc{Disconnected Matching} is the same as \textsc{Maximum Matching}, which is known to be polynomial-time solvable \cite{micali}. Regarding disconnected matchings, \textsc{Disconnected Matching} is $\mathsf{NP}$-$\mathsf{complete}$ for chordal graphs \cite{gomes} and polynomial-time solvable for interval graphs \cite{gomes}. 

\medskip

From the viewpoint of Parameterized Complexity, Gomes et al. \cite{gomes} proved that for graphs with a polynomial number of minimal separators, \textsc{Disconnected Matching} parameterized by the number of connected components, belongs to the class $\mathsf{XP}$. Furthermore, unless $\mathsf{NP}\subseteq \mathsf{coNP/poly}$, \textsc{Disconnected Matching} does not admit a polynomial kernel when parameterized by the vertex cover number plus the size of the matching nor when parameterized by the vertex deletion
distance to clique plus the size of the matching. In the same paper (\cite{gomes}), the authors also proved that \textsc{Disconnected Matching} is \FPT when parameterized by treewidth (\tw). They used the standard dynamic programming technique, and the running time of their algorithm is $\mathcal{O}(8^{\tw}\cdot \eta_{\tw+1}^{3}\cdot n^{2})$, where
$\eta_{i}$ is the $i$-th Bell number\footnote{The \emph{Bell number} $\eta_{i}$ counts the number of different ways to partition a set that has exactly $i$ elements. Mathematically,
$\eta_{i+1}=\displaystyle\sum_{k=0}^i{i \choose k}\eta_{k}.$
}. Further, we mention the following proposition, which is an immediate consequence of the fact that \textsc{Induced Matching} is a special case of \textsc{Disconnected Matching}.

 \begin{proposition} [\cite{gomes}] \label{relatedobs}
\textsc{Disconnected Matching} is $\mathsf{NP}$-$\mathsf{complete}$ for every graph class for which \textsc{Induced Matching} is $\mathsf{NP}$-$\mathsf{complete}$.
\end{proposition}

\section{Main Results}
\label{sec:main}

In this paper, we consider the parameter to be the treewidth of the input graph, and as is customary in the field, we suppose that the input also consists of a tree decomposition $\mathcal{T} = (\mathbb{T},\{\mathcal{B}_{x}\}_{x\in V(\mathbb{T})})$ of width \tw~of the input graph.

First, in Section \ref{IM}, we present a $3^{\tw}\cdot \tw^{\mathcal{O}(1)}\cdot n$ time algorithm for \textsc{Induced Matching}, improving upon the $\mathcal{O}(4^{\tw}\cdot n)$ time bound by Moser and Sikdar \cite{moser}. For this purpose, we use a nice tree decomposition that satisfies the ``deferred edge property'' (defined in Section \ref{PC}) and the fast subset convolution (see Section \ref{AD}) for the join nodes.

\begin{restatable}{theorem}{inducedmatching}\label{im}
\textsc{Induced Matching} can be solved in $3^{\tw}\cdot \tw^{\mathcal{O}(1)}\cdot n$ time by a deterministic algorithm.
\end{restatable}

In Section \ref{AM}, we present a $6^{\tw} \cdot n^{\mathcal{O}(1)}$ time algorithm for \textsc{Acyclic Matching}, improving the result by Hajebi and Javadi \cite{hajebi}, who proved that \textsc{Acyclic Matching} parameterized by $\tw$ is $\mathsf{FPT}$. They used Courcelle's theorem, which is purely theoretical, and thus the hidden parameter dependency in the running time is huge (a tower of exponents). To develop our algorithm, we use the \emph{Cut $\&$ Count} method introduced by Cygan et al. \cite{cygan} in addition to the fast subset convolution. The Cut $\&$ Count method allows us to deal with connectivity-type problems through randomization; here, randomization arises from the usage of the \textit{Isolation Lemma} (see Section \ref{PC}). 

 \begin{restatable} {theorem}{acyclicmatching} \label{am}
	 $\textsc{Acyclic Matching}$ can be solved in $6^{\tw}\cdot n^{\mathcal{O}(1)}$ time by a randomized algorithm. The algorithm cannot give false positives and may
	give false negatives with probability at most $\frac{1}{3}$. 
\end{restatable}

In Section \ref{CDM}, we present a $(3c)^{\tw}\cdot  \tw^{\mathcal{O}(1)}\cdot n$ time algorithm for $c$-\textsc{Disconnected Matching}. We use the dynamic programming technique along with the fast subset convolution for the join nodes. This resolves an open question by Gomes et al. \cite{gomes}, who asked whether $c$-\textsc{Disconnected Matching} can be solved in a single exponential time with vertex cover ($\rm{vc}$) as the parameter. Since for any graph $G$, $\tw(G) \leq \rm{vc}(G)$, we answer their question in the affirmative. 

\begin{restatable}{theorem}{cdisconnectedmatching}\label{cdm}
 For a fixed positive integer $c\geq 2$, $c$-\textsc{Disconnected Matching} can be solved in $(3c)^{\tw}\cdot  \tw^{\mathcal{O}(1)}\cdot n$ time by a deterministic algorithm.
\end{restatable}

In Section \ref{DM}, we present a lower bound for the time complexity of \textsc{Disconnected Matching}, proving that for any choice of a constant $c$, an $\mathcal{O}(c^{\tw}\cdot n)$-time algorithm for the \textsc{Disconnected Matching} problem is unlikely. In fact, we prove that even an $\mathcal{O}(c^{\pw}\cdot n)$-time algorithm is not possible, where $\pw$ is the \emph{pathwidth} (see Section \ref{PC}) of the graph which is bounded from below by the treewidth.

\begin{restatable}{theorem}{disconnectedmatching} \label{dm} Assuming the Exponential Time Hypothesis to be true, there is no $2^{o(\pw \log \pw)}\cdot n^{\mathcal{O}(1)}$-time algorithm for \textsc{Disconnected Matching}. 
\end{restatable}

Also, we briefly discuss the \textsf{SETH} lower bounds in the Conclusion.

\section{Preliminaries}\label{prelim}

\subsection{Graph-theoretic Notations and Definitions}\label{GT}

For a graph $G$, let $V(G)$ denote its vertex set and $E(G)$ denote its edge set. Given a matching $M$, a vertex
$v\in V(G)$ is \textit{$M$-saturated} if $v$ is incident on an edge of $M$. Given a graph $G$ and a matching $M$, let $V_{M}$ denote the set of $M$-saturated vertices and $G[V_{M}]$ denote the
subgraph of $G$ induced by $V_{M}$. The
\textit{matching number} of $G$ is the maximum cardinality of a matching among all matchings in $G$, and we denote it by $\mu(G)$. A matching that saturates all
the vertices of a graph is a \textit{perfect matching}. If $uv\in M$, then
$v$ is the \textit{$M$-mate} of $u$, and vice versa. 
		The $\mathcal{P}$-\emph{matching number} of $G$ refers to the maximum cardinality of a $\mathcal{P}$-matching among all $\mathcal{P}$-matchings in $G$. We denote by $\mu_{\mathsf{induced}}(G)$, $\mu_{\mathsf{acyclic}}(G)$, $\mu_{c\mathsf{,discon}}(G)$, and $\mu_{\mathsf{disconnected}}(G)$, the \emph{induced matching number}, the \emph{acyclic matching number}, the \emph{$c$-disconnected matching number}, and the \emph{disconnected matching number} of $G$, respectively. It is worth noting that in any $c$-disconnected matching, $c$ can be at most $\mu_{\mathsf{induced}(G)}$. Furthermore, the following proposition outlines the relationship among various $\mathcal{P}$-matching numbers.

\begin{proposition} [\cite{gomes}] \label{prelimobs} For a graph $G$, the following hold:
\begin{enumerate}
    \item $\mu(G)=\mu_{1\mathsf{,discon}}(G)\geq \mu_{2,\mathsf{discon}}(G) \geq \ldots \geq \mu_{\mu_{\mathsf{induced}(G)}\mathsf{,discon}}(G)\geq \mu_{\mathsf{induced}}(G).
	$
	\item 	$\mu(G)\geq \mu_{\mathsf{acyclic}}(G)\geq \mu_{\mathsf{induced}}(G). $
\end{enumerate}
	\end{proposition}

For a vertex set $X\subseteq V(G)$, $G[X]$ denotes the subgraph induced by $X$. A \emph{cut} of a set $X \subseteq V(G)$ is a pair $(X_{l},X_r)$ with $X_{l}\cap X_{r}=\emptyset$ and $X_{l}\cup X_{r}=X$, where $X$ is an arbitrary subset of $V(G)$. When $X$ is immaterial, we do not mention it explicitly. A cut $(X_l,X_r)$ is \emph{consistent} in a subgraph $H$ of $G$ if $u\in X_{l}$ and $v\in X_{r}$ implies $uv\notin E(H)$. A \textit{forest} is an undirected graph in which any two vertices are connected by at most one simple path (a path that does not have repeating vertices). For a graph $G$, let $cc(G)$ denote the number of connected components of $G$. For an undirected graph $G$, the \emph{open neighborhood} of a vertex $v$, denoted by $N(v)$, stands for $\{u\in V(G):uv\in E(G)\}$, while the \emph{closed neighborhood} of $v$ is $N[v]= N(v) \cup \{v\} $. Standard graph-theoretic terms not explicitly defined here can be found in \cite{diestel}.

Let $G$ be a graph. A \emph{coloring} on a set $X\subseteq V(G)$ is a function $f:X\rightarrow S$, where $S$ is any set. Here, the elements of $S$ are called \emph{colors}. A coloring defined on an empty set is an \emph{empty coloring}. For a coloring $f$ on $X\subseteq V(G)$ and $Y\subseteq X$, we use the notation $f|_{Y}$ to denote the restriction of $f$ to $Y$. For a coloring $f:X(\subseteq V(G)) \rightarrow S$, a vertex $v\in V(G)\setminus X$, and a color
$\alpha\in S$, we define $f_{v\rightarrow \alpha}:X \cup \{v\}\rightarrow S$ as follows:
$$
  f_{v\rightarrow \alpha}(x) = \begin{cases}
        f(x) \ \ \ \ \text{\ if \ $x\in X$,}
        \\
        \alpha \ \ \  \ \ \ \ \    \text{\ if \ $x=v$.}
           
        \end{cases}
$$

More generally, for a coloring $f:X(\subseteq V(G))\rightarrow \{0,1,2\}$, a set $Y\subseteq X$, and a color
$\alpha\in \{0,1,2\}$, we define $f_{Y\rightarrow \alpha}:X \rightarrow \{0,1,2\}$ as follows:

$$
  f_{Y\rightarrow \alpha}(x) = \begin{cases}
        f(x) \ \ \ \ \text{\ if \ $x\notin Y$,}
        \\
        \alpha \ \ \  \ \ \ \ \    \text{\ if \ $x\in Y$.}
           
        \end{cases}
$$

\begin{definition} [Correct Coloring] \label{prelimdef}
Given a graph $G$ and a set $X\subseteq V(G)$, two colorings $f_{1},f_{2}: X \rightarrow \{0,1,2\}$ are \emph{correct} for a coloring $f: X \rightarrow \{0,1,2\}$ if the following conditions hold:
\begin{enumerate}
	\item  $f(v)=0$ if and only if $f_{1}(v)=f_{2}(v)=0$,
	\item  $f(v)=1$ if and only if  $(f_{1}(v),f_{2}(v))\in\{(1,2),(2,1)\}$, and
	\item  $f(v)=2$ if and only if $f_{1}(v)=f_{2}(v)=2$.
\end{enumerate}
\end{definition}

 \subsection{ Algebraic Definitions}\label{AD}
For a set $X$, let $2^{X}$ denote the set of all subsets of $X$. For a positive integer $k$, let  $[k]$ denote the set $\{1,\ldots,k\}$. In the set $[k]\times [k]$, a \emph{row} is a set $\{i\}\times [k]$ and a \emph{column} is a set $[k]\times \{i\}$ for some $i\in [k]$. For two integers, $a$ and $b$, we use $a\equiv b$ to indicate
that $a$ is even if and only if $b$ is even. If $\mathsf{w}:U\rightarrow\{1,\ldots,N\}$, then for $S\subseteq U$, $\mathsf{w}(S)=\displaystyle{\sum_{e\in S}}\mathsf{w}(e)$. For definitions of \emph{ring} and \emph{semiring}, we refer the readers to any elementary book on abstract algebra. Given an integer $n > 1$, called a \emph{modulus}, two integers \textit{a} and \textit{b} are \emph{congruent modulo} $n$ if there is an integer $k$ such that $a-b = kn$. Note that two integers are said to be \textit{congruent modulo $2$} if they have the same parity (that is, either both are odd or both are even). For a set $S$, we use the notation $|S|_{2}$ to denote the number of elements in set $S$ congruent modulo $2$. We remark that in formulas, it is more convenient to use this notation than the phrase ``parity".

\medskip

\noindent \textbf{Subset Convolution is defined as follows.}

\begin{definition}\label{defsubsetconvo}
Let $B$ be a finite set and $\mathbb{R}$ be a semiring. Then, the \emph{subset convolution} of two functions $f,g:2^{B}\rightarrow \mathbb{R}$ is the function
$(f\ast g):2^{B}\rightarrow \mathbb{R}$ such that for every $Y \subseteq B$,
\begin{equation} \label{eq}
(f\ast g)(Y)=\displaystyle \sum_{X\subseteq Y } f(X)g(Y\setminus X).
\end{equation}
\end{definition}

Equivalently, (\ref{eq}) can be written as 
\begin{equation}
(f*g)(Y)=\displaystyle \sum_{\substack{A \cup B =Y \\ A \cap B =\emptyset }} f(A)g(B).
\end{equation}

Given $f$ and $g$, a direct evaluation of $f\ast g$ for all $X\subseteq Y$ requires $\Omega(3^{n})$ semiring operations, where $n=|B|$. However, we have the following result by Br\"ojrklund et al. \cite{bjorlund}. 
\begin{proposition} [\cite{bjorlund,cygan}] \label{B1}
For two functions $f,g : 2^B \rightarrow \mathbb{Z}$, where $n=|B|$ and $\mathbb{Z}$ is a ring, given all the $2^n$ values of $f$ and $g$ in the input, all the $2^n$ values of the subset convolution of
$f*g$ can be computed in $\mathcal{O}(2^n \cdot n^3)$ arithmetic operations.
\end{proposition}

If the input
functions have an integer range $\{-P,\ldots,P\}$, their subset convolution over the ordinary sum-product ring (see \cite{bjorlund, cygan} for definition) can be computed in $\widetilde{\mathcal{O}}(2^n\log P)$ time \cite{bjorlund}. However, in many problems, we want the subset convolution over the max-sum semiring (see \cite{bjorlund, cygan} for definition), i.e.,
the semiring $(\mathbb{Z}\cup \{-\infty\}, \max, +)$. Note that in the max-sum semiring, the role of the $+$ operation changes: in the max-sum semiring, $+$ plays the role of the multiplicative operation. While the fast subset convolution algorithm does not directly apply to semirings where additive inverses need not exist, one can, fortunately, embed the integer max-sum semiring into the integer sum-product ring (see \cite{bjorlund}). Thus, we have the following result.
\begin{proposition} [\cite{cygan}] \label{B2}
For two functions $f,g : 2^B \rightarrow \{-P,\ldots,P\}$, where $n=|B|$, given all the
$2^n$ values of $f$ and $g$ in the input, all the $2^n$ values of the subset convolution of
$f*g$ over the integer max-sum semiring can be computed in time $2^{n}\cdot n^{\mathcal{O}(1)}\cdot \mathcal{O}(P \log P  \log \log P)$.
\end{proposition}

\subsection{Parameterized Complexity Definitions}	\label{PC}
 In the framework of Parameterized Complexity, each instance of a problem $\mathrm{\Pi}$ is associated with a non-negative integer \textit{parameter} $k$. A parameterized problem $\mathrm{\Pi}$ is \textit{fixed-parameter tractable} ($\mathsf{FPT}$) if there is an algorithm that, given an instance, $(I,k)$ of $\mathrm{\Pi}$, solves it in time $f(k)\cdot |I|^{\mathcal{O}(1)}$, for some computable function $f(\cdot)$. Central to Parameterized
Complexity is the hierarchy of complexity classes, which is defined as follows:
\begin{equation} \label{fpt}
\mathsf{FPT} \subseteq \mathsf{W[1]}  \subseteq \mathsf{W[2]} \subseteq \ldots \subseteq \mathsf{XP}.
\end{equation}

All inclusions in (\ref{fpt}) are believed to be strict. In particular, \FPT $\neq$ $\mathsf{W[1]}$ under the Exponential Time Hypothesis (defined below).
Here, the class $\mathsf{W[1]}$ is the analog of \NP in Parameterized Complexity.

To obtain (essentially) tight conditional lower bounds for the running times of algorithms,
we rely on the well-known \textit{Exponential Time Hypothesis} ($\mathsf{ETH}$). To formalize the statement of $\mathsf{ETH}$, first recall that
given a formula $\psi$ in the conjunctive normal form ($\mathsf{CNF}$) with $n$ variables and $m$ clauses, the task of CNF-SAT is to decide whether there is a truth assignment to the variables that satisfies
$\psi$. In the $q$-CNF-SAT problem, each clause is restricted to have at most $q$ literals. 
\ETH asserts that 3-CNF-SAT cannot be solved in time $\mathcal{O}(2^{o(n)})$ while \textsf{SETH} asserts that for every $\epsilon >0$, there is a constant $q$
such that $q$-CNF-SAT on $n$ variables cannot be solved in time $(2-\epsilon)^n\cdot n^{\mathcal{O}(1)}$ \cite{imp}. More information on Parameterized Complexity, \textsf{ETH}, and \textsf{SETH} can be found in \cite{cygan, downey}.

A parameterized (decision) problem $\mathrm{\Pi}$ is said to admit a \emph{kernel} of size $f(k)$ for some function $f$ that depends only on $k$ if the following is true: There exists an algorithm (called a \emph{kernelization algorithm}) that runs in $(|I|+k)^{\mathcal{O}(1)}$ time and translates any input instance $(I,k)$ of $\mathrm{\Pi}$ into an equivalent instance $(I',k')$ of $\mathrm{\Pi}$ such that the size of $(I',k')$ is bounded by $f(k)$. If the function $f$ is a polynomial, then the problem is said to admit a \emph{polynomial kernel}. It is well-known that a decidable parameterized problem is $\mathsf{FPT}$ if and only if it has
a kernel. Note that if the parameterized problem is solvable in time $f(k)\cdot n^c$ for some $f$ and $c$, then the proof yields a kernel of size $f(k)$. Standard notions in Parameterized Complexity not explicitly defined here can be found in \cite{cygan}.

\begin{definition} [Equivalent Instances] \label{equivalent}
Let $\mathrm{\Pi}_{1}$ and $\mathrm{\Pi}_{2}$ be two parameterized problems. Two instances, $(I, k)\in$ $\mathrm{\Pi}_{1}$ and $(I', k')\in $ $\mathrm{\Pi}_{2}$, are \emph{equivalent} when $(I, k)$ is a Yes-instance if and only if $(I', k')$ is a Yes-instance.

\end{definition}
\begin{definition} [Monte Carlo Algorithms with False Negatives] \label{definitionmonte}
An algorithm is a \emph{Monte Carlo algorithm with false negatives} if it satisfies the following property when asked about the existence of an object: If it answers yes, then it is true, and if it answers no, then it is correct with probability at least $\frac{2}{3}$ $($here, the constant $\frac{2}{3}$ is chosen arbitrarily$)$.
\end{definition}

\noindent \textbf{Cut \& Count Method.}
The Cut \& Count method was introduced by Cygan et al. \cite{cygan1}. It is a tool for designing algorithms with a single exponential running time for problems with certain connectivity requirements. The method is broadly divided into the following two parts. 
\begin{itemize}
\item \textbf{The Cut part:} Let $\mathcal{S}$ denote the set of feasible solutions. Here, we relax the connectivity requirement by considering a set $\mathcal{R}$ that contains feasible candidate solutions, which may or may not be connected. Furthermore, we consider a set $\mathcal{C}$ of pairs $(X,C)$, where $X\in \mathcal{R}$ and $C$ is a consistent cut of $X$.
\item \textbf{The Count part:} Here, we compute the cardinality of $\mathcal{C}$ modulo 2 (see Section \ref{AD}) using a sub-procedure. Non-connected candidate solutions $X\in \mathcal{R} \setminus \mathcal{S}$
cancel since they are consistent with an even number of cuts.  Only connected candidates $x\in \mathcal{S}$ are retained for further consideration.
\end{itemize}

More information on the Cut \& Count method can be found in \cite{cygan,cygan1}.
\smallskip

\noindent \textbf{Isolation Lemma.} \label{sectionisolationlemma} Consider the following definition.

\begin{definition}\label{defiso}
Let $U$ be a universe.	A function $\mathsf{w}:U\rightarrow \mathbb{Z}$ isolates a set family $\mathcal{F}\subseteq 2^{U}$ if there is a unique $S'\in \mathcal{F}$ with $\mathsf{w}(S')=\displaystyle\min_{S\in \mathcal{F}}\mathsf{w}(S)$.
\end{definition}
\begin{lemma} [Isolation Lemma, \cite{mulmuley}]  \label{IL1}
Let $\mathcal{F}\subseteq 2^{U}$ be a set family over a universe $U$ with $ \mathcal{|F|}>0$. For each $u\in U$, choose a weight $\mathsf{w}(u)\in \{1,2,\ldots, N\}$ uniformly and independently at random. Then
$$ \mathsf{prob}(\mathsf{w} \ isolates \ \mathcal{F})\geq 1- \frac{|U|}{N}.$$
\end{lemma}

\noindent \textbf{Treewidth.} A \emph{rooted tree} $T$ is a tree having a distinguished vertex labeled $r$, called the \emph{root}. For a vertex, $v\in V(T)$, an \emph{$(r, v)$-path} in $T$ is a sequence of distinct vertices starting from $r$ and ending at $v$ such that every two consecutive vertices are connected by an edge in the tree. The \emph{parent} of a vertex $v$ different from $r$ is its
neighbor on the unique $(r, v)$-path in $T$. The other neighbors of $v$ are its \emph{children}. A vertex $u$ is an \emph{ancestor} of $v$ if $u\neq v$ and $u$ belongs on
the unique $(r, v)$-path in $T$. A \emph{descendant} of $v$ is any vertex $u\neq v$ such that $v$ is its ancestor. The subtree rooted at $v$ is the subgraph of $T$ induced by $v$ and its descendants.

\begin{definition} [Tree Decomposition] 
A \emph{tree decomposition} of a graph $G$ is a pair $\mathcal{T}=(\mathbb{T},\{\mathcal{B}_{x}\}_{x\in V(\mathbb{T})})$, where $\mathbb{T}$ is a tree and each $\mathcal{B}_{x},x\in V(\mathbb{T})$, is a subset of $V(G)$ called a \emph{bag}, such that 
\begin{enumerate}
\item [$T.1)$] $\bigcup_{x\in V(\mathbb{T})}\mathcal{B}_{x}=V(G)$,
\item [$T.2)$] for any edge $uv\in E(G)$, there exists a node $x\in V(\mathbb{T})$ such that $u,v\in \mathcal{B}_{x}$, 
\item [$T.3)$] for all $x,y,z\in V(\mathbb{T})$, if $y$ is on the path from $x$ to $z$ in $\mathbb{T}$ then $\mathcal{B}_{x}\cap \mathcal{B}_{z}\subseteq \mathcal{B}_{y}$.
\end{enumerate}
\end{definition}

The \emph{width} of a tree decomposition $\mathcal{T}$ is the size of its largest bag minus one. The \emph{treewidth} of $G$ is the \emph{minimum width} over all tree decompositions of $G$. We denote the treewidth of a graph by $\tw$.

Dynamic programming algorithms on tree decompositions are often presented on \emph{nice tree decompositions}, which
were introduced by Kloks \cite{kloks}. We refer to the tree decomposition definition given by Kloks as a standard nice tree
decomposition, which is defined as follows: 

\begin{definition} [Nice Tree Decomposition] 
	Given a graph $G$, a tree decomposition 
$\mathcal{T} = (\mathbb{T},\{\mathcal{B}_{x}\}_{x\in V(\mathbb{T})})$ of $G$ is a \emph{nice tree decomposition} if the following hold:
\begin{enumerate}
	\item [$N.1)$] $\mathcal{B}_{r}=\emptyset$, where $r$ is the root of $\mathbb{T}$,
	and $\mathcal{B}_{l}=\emptyset$ for every leaf $l$ of $\mathbb{T}$.
	\item [$N.2)$] Every non-leaf node $x$ of $\mathbb{T}$ is of one of the following types:
	\begin{enumerate}
		\item [$N.2.1)$]
		\noindent \textbf{Introduce vertex node:} $x$ has exactly one child $y$, and $\mathcal{B}_{x}=\mathcal{B}_{y}\cup \{v\}$ where $v\notin \mathcal{B}_{y}$. We say that $v$ is \emph{introduced} at $x$.
		
		\item [$N.2.2)$] \noindent \textbf{Forget vertex node:} $x$ has exactly one child $y$, and $\mathcal{B}_{x}=\mathcal{B}_{y}\setminus \{u\}$ where $u\in \mathcal{B}_{y}$. We say that $u$ is \emph{forgotten} at $x$.
		
		\item [$N.2.3)$] \noindent \textbf{Join node:}  $x$ has exactly two children, $y_{1}$ and $y_{2}$, and $\mathcal{B}_{x}=\mathcal{B}_{y_{1}}=\mathcal{B}_{y_{2}}$.
	\end{enumerate}	
	\end{enumerate}
\end{definition}

\begin{observation} \label{joinobs}
All the common nodes in the bags of the subtrees of children of a join node appear in the bag of the join node.
\end{observation}

For our problems, we want the standard nice tree decomposition to satisfy an additional property, and that is, among the vertices present in the bag of a join node, no edges have been introduced yet.  To achieve this, we use another known
type of node, an \emph{introduce edge node}, which is defined as follows:\\

	 \noindent \textbf{Introduce edge node:} $x$ has exactly one child $y$, and $x$ is labeled with an edge $uv\in E(G)$ such that $u,v\in \mathcal{B}_{x}$ and $\mathcal{B}_{x}=\mathcal{B}_{y}$. We say that $uv$ is \emph{introduced} at $x$.
	 
The use of introduce edge nodes enables us to add edges one by one in our nice tree decomposition. We additionally require that every edge is introduced exactly once. Observe that condition $T.3)$ implies that, in a nice tree decomposition, for every vertex $v \in V(G)$, there
exists a unique highest node $t(v)$ such that $v\in \mathcal{B}_{t(v)}$. Moreover, the parent
of $\mathcal{B}_{t(v)}$ is a forget node that forgets $v$. Consider an edge $uv\in E(G)$, and
observe that $T.2)$ implies that either $t(v)$ is an ancestor of $t(u)$ or $t(u)$ is an ancestor
of $t(v)$. Without loss of generality, assume the former, and observe that the introduce edge bag that introduces $uv$ can be inserted anywhere between $t(u)$ and its
parent (which forgets $u$). So, for every edge $uv$, where $t(v)$ is an ancestor of $t(u)$, we want our nice tree decomposition to insert the introduce edge bags (introducing edges of the form $uv$) between $t(u)$ and its parent in an arbitrary order. If a nice tree decomposition having introduce edge nodes satisfies these additional conditions, then we say that it exhibits the \textit{deferred edge} property (informally speaking, as we are deferring the introduction of edges in our nice tree decomposition to as late as possible).

Given a tree decomposition of a graph $G$, where $n=|V(G)|$, a standard nice tree decomposition of equal width and at most $\mathcal{O}(\tw\cdot n)$ nodes can be found in $\tw^{\mathcal{O}(1)}\cdot n$ time \cite{kloks}, and in the same running time, a standard nice tree decomposition can
be easily transformed to the variant satisfying the \textit{deferred edge} property as follows: Recall that we may insert the introduce edge bag that introduces $uv$
between $t(u)$ and its parent (which forgets $u$). This transformation, for every edge $uv\in E(G)$,
can be easily implemented in time $\tw^{\mathcal{O}(1)}\cdot n$ by a single top-down transversal of the tree decomposition. It is worth noting that the resulting tree decomposition will still have $\mathcal{O}(\tw\cdot n)$ nodes, as a graph with treewidth at most \tw~has at most $\tw\cdot n$ edges, as mentioned in \cite{cygan}. 

For each node $x$ of the tree decomposition, let $V_{x}$ be the union of all the bags  present in the subtree of $\mathbb{T}$ rooted at $x$, including $\mathcal{B}_{x}$.
For each node $x$ of the tree decomposition, define the subgraph $G_{x}$ of $G$ as follows:

$G_{x}=(V_{x},E_{x}=\{e: e$ is introduced in the subtree of $\mathbb{T}$ rooted at $x\})$.

A \emph{path decomposition} and \emph{pathwidth} are defined analogously as tree decomposition and treewidth with the additional requirement that the tree $\mathbb{T}$ is a \emph{path}. The 
pathwidth of a graph $G$ is denoted by $\pw(G)$, and when there is no confusion, we use only $\pw$ to denote $\pw(G)$.
\medskip

In Sections \ref{IM} to \ref{CDM}, we use different colors to represent the possible states of a vertex in a bag $\mathcal{B}_{x}$ of $\mathcal{T}$ with respect to a matching $M$ as follows: 
\begin{itemize}
\item  \textbf{white(0)}: A vertex colored 0 is not saturated by $M$.

\item  \textbf{black(1)}: A vertex colored $1$ is saturated by $M$, and the edge between the vertex and its $M$-mate has also been introduced in $G_{x}$.

\item \textbf{gray(2)}: A vertex colored $2$ is saturated by $M$, and either its $M$-mate has not yet been introduced in $G_{x}$, or the edge between the vertex and its $M$-mate has not yet been introduced in $G_{x}$.
\end{itemize}

Now, consider the following definition used in Sections \ref{AM} and \ref{CDM}.

\begin{definition} [Valid Coloring] \label{valid} Given a node $x$ of $\mathbb{T}$, a coloring $d:\mathcal{B}_{x}\rightarrow\{0,1,2\}$ is \emph{valid} on $\mathcal{B}_{x}$ if there exists a coloring $\widehat{d}:V_{x}\rightarrow\{0,1,2\}$ in $G_{x}$, called a \emph{valid extension} of $d$, such that the following hold: 

\begin{enumerate}
	\item [$(i)$] $\widehat{d}$ restricted to $\mathcal{B}_{x}$ is exactly $d$.
	\item  [$(ii)$] The subgraph induced by the vertices colored 1 under $\widehat{d}$ has a perfect matching.
		\item [$(iii)$] Vertices colored $2$ under $\widehat{d}$ must all belong to $\mathcal{B}_{x}$.
\end{enumerate}
\end{definition}

\medskip

\section{Algorithm for Induced Matching} \label{IM}
 In this section, we present a $3^{\tw}\cdot \tw^{\mathcal{O}(1)}\cdot n$-time algorithm for \textsc{Induced Matching} assuming that we are given a nice tree decomposition
$\mathcal{T} = (\mathbb{T},\{\mathcal{B}_{x}\}_{x\in V(\mathbb{T})})$ of $G$ of width $\tw$ that satisfies the \textit{deferred edge} property. For this purpose, we define the following notion.

\begin{definition} [Valid Induced Coloring] \label{im1}
Given a node $x$ of $\mathbb{T}$, a coloring $f:\mathcal{B}_{x}\rightarrow\{0,1,2\}$ is \emph{valid induced} if there exists a coloring $\widehat{f}:V_{x}\rightarrow\{0,1,2\}$ in $G_{x}$, called a \emph{valid induced extension} of $f$, such that the following hold: 

\begin{enumerate}
\item  $\widehat{f}$ restricted to $\mathcal{B}_{x}$ is exactly $f$.
\item The subgraph induced by the vertices colored 2 under $\widehat{f}$ is a set of isolated vertices. Furthermore, vertices colored $2$ under $\widehat{f}$ must all belong to $\mathcal{B}_{x}$.
\item  The subgraph induced by the vertices colored 1 under $\widehat{f}$ is an induced matching.
\end{enumerate}
\end{definition}

We have a table $\mathcal{M}$ with an entry $\mathcal{M}_{x}[f]$ for each node $x$ of $\mathbb{T}$ and for every coloring $f:\mathcal{B}_{x}\rightarrow\{0,1,2\}$. Note that we have at most $\mathcal{O}(\tw\cdot n)$ many choices for $x$ and at most $3^{\tw}$ many choices for $f$. Therefore, the size of table $\mathcal{M}$ is bounded by $\mathcal{O}(3^{\tw}\cdot \tw\cdot 
n)$. The following definition specifies the value each entry $\mathcal{M}_{x}[f]$ of $\mathcal{M}$ is supposed to store.
\begin{definition}
If $f$ is valid induced, then the entry $\mathcal{M}_{x}[f]$ stores the maximum number of vertices that are colored $1$ or $2$ under some valid induced extension $\widehat{f}$ of $f$ in $G_{x}$. Else, the entry $\mathcal{M}_{x}[f]$ stores the value $-\infty$ and marks $f$ as \emph{invalid}.
\end{definition}

Since the root of $\mathbb{T}$ is an empty node, note that the maximum number of vertices saturated by any induced matching is exactly $\mathcal{M}_{r}[\emptyset]$, where $r$ is the root of the decomposition $\mathcal{T}$. 

We now provide recursive formulas to compute the entries of table $\mathcal{M}$. 
\bigskip

\noindent \textbf{Leaf node:} For a leaf node $x$, we have that $\mathcal{B}_{x}=\emptyset$. Hence there is only one possible coloring on $\mathcal{B}_{x}$, that is, the empty coloring, and we have $\mathcal{M}_{x}[\emptyset]=0$. 

\medskip
    
\noindent \textbf{Introduce vertex node:} Suppose that $x$ is an introduce vertex node with child node $y$ such that $\mathcal{B}_{x}=\mathcal{B}_{y}\cup \{v\}$ for some $v\notin \mathcal{B}_{y}$. Note that we have not introduced any edges incident on $v$ so far, so $v$ is isolated in $G_x$. For every coloring $f: \mathcal{B}_{x}\rightarrow \{0,1,2\}$, we have the following recursive formula:

\begin{equation*}
  \mathcal{M}_{x}[f] = \begin{cases}
        \mathcal{M}_{y}[f|_{\mathcal{B}_{y}}] \ \  \ \ \ \ \ \ \text{\ if \ $f(v)=0$,}
        \\
        -\infty \ \ \  \ \ \ \ \ \ \ \ \ \ \ \  \text{\ if \ $f(v)=1$,}
            \\
            \mathcal{M}_{y}[f|_{\mathcal{B}_{y}}]+1 \ \ \ \text{\ if \ $f(v)=2$}.
        \end{cases}
 \end{equation*}

\bigskip

Note that when $f(v)=1$, then $f$ is invalid as $v$ does not have any neighbor in $G_{x}$ (by the definition of a valid induced coloring, $v$ needs one neighbor of color $1$ in $G_{x}$), and hence $\mathcal{M}_{x}[f]=-\infty$. Next, when $f(v)=0$ or $f(v)=2$, then $f$ is valid induced if and only if $f|_{\mathcal{B}_{y}}$ is valid induced. Moreover, when $f(v)=2$, we increment the value by one as one more vertex has been colored $2$ in $G_{x}$.

Clearly, the evaluation of all introduce vertex nodes can be done in $3^{\tw}\cdot \tw^{\mathcal{O}(1)} \cdot n$ time.
\bigskip

\noindent \textbf{Introduce edge node:} 
Suppose that $x$ is an introduce edge node that introduces
an edge $uv$, and let $y$ be the child of $x$.  For every coloring $f: \mathcal{B}_{x}\rightarrow \{0,1,2\}$, we consider the following cases:

If at least one of $f(u)$ or $f(v)$ is $0$, then 
	\begin{equation*}
	\mathcal{M}_{x}[f]=
\mathcal{M}_{y}[f]. 
\end{equation*}

Else, if $f(u)=f(v)=1$, then
\begin{equation*}
	\mathcal{M}_{x}[f]=
\mathcal{M}_{y}[f_{\{u,v\}\rightarrow 2}].
\end{equation*}

Else, $\mathcal{M}_{x}[f]=-\infty$.

If either $f(u)$ or $f(v)$ is $0$, then $f$ is valid induced if and only if $f$ is valid induced on $\mathcal{B}_{y}$. Next, let us consider the case when both $f(u)$ and $f(v)$ are $1$ and $f$ is valid induced. In this case, both $u$ and $v$ must be colored $2$ under $f$ in $\mathcal{B}_{y}$ (this follows by the definition of a valid induced coloring).

Clearly, the evaluation of all introduce edge nodes can be done in $3^{\tw}\cdot \tw^{\mathcal{O}(1)} \cdot n$ time.
  
\bigskip

\noindent  \textbf{Forget node:} Suppose that $x$ is a forget vertex node with
a child $y$ such that $\mathcal{B}_{x}=\mathcal{B}_{y}\setminus \{u\}$ for some $u\in \mathcal{B}_{y}$. For every coloring $f: \mathcal{B}_{x}\rightarrow \{0,1,2\}$, we have
 
\begin{equation} \label{forgetim}\mathcal{M}_{x}[f]=
\displaystyle\max\{\mathcal{M}_{y}[f_{u \rightarrow 0}],\mathcal{M}_{y}[f_{u \rightarrow 1}]\} .
\end{equation}

 The first term on the right-hand side in (\ref{forgetim}) corresponds to the case when $f(u)=0$ in $\mathcal{B}_{y}$, and the second term corresponds to the case when $f(u)=1$ in $\mathcal{B}_{y}$. Note that the maximum is taken over colorings $f_{u \rightarrow 0}$  and $f_{u \rightarrow 1}$  only, as the coloring $f_{u \rightarrow 2}$ cannot be extended to a valid induced coloring once $u$ is forgotten. 

Clearly, the evaluation of all forget nodes can be done in $3^{\tw}\cdot \tw^{\mathcal{O}(1)} \cdot n$ time.

\bigskip

\noindent \textbf{Join node:}
Let $x$ be a join node with children $y_{1}$ and $y_{2}$. For every coloring $f: \mathcal{B}_{x}\rightarrow \{0,1,2\}$, we have 
\begin{equation*}
\mathcal{M}_{x}[f]=\displaystyle\max_{f_{1},f_{2}} \{\mathcal{M}_{y_{1}}[f_{1}]+\mathcal{M}_{y_{2}}[f_{2}]-|f^{-1}(1)|-|f^{-1}(2)| \},
\end{equation*}
where $f_{1}: \mathcal{B}_{y_{1}}\rightarrow \{0,1,2\}$ and $f_{2}: \mathcal{B}_{y_{2}}\rightarrow \{0,1,2\}$ such that $f_{1}$ and $f_{2}$ are correct for $f$ (see Definition \ref{prelimdef}).
 
Note that we are determining the value of $\mathcal{M}_{x}[f]$ by looking up the corresponding coloring in nodes $y_{1}$ and $y_{2}$, adding the corresponding values, and subtracting the number of vertices
colored $1$ or $2$ under $f$. Note that the subtraction is necessary; otherwise, by Observation \ref{joinobs}, the number of vertices colored $1$ or $2$ in $\mathcal{B}_{x}$ would be counted twice.

By the naive method, the evaluation of all join nodes altogether can be done in $4^{\tw}\cdot \tw^{\mathcal{O}(1)} \cdot n$ time as follows.

For a given coloring $f: \mathcal{B}_{x} \rightarrow \{0,1,2\}$, where $|\mathcal{B}_{x}|=p$, $1\leq p \leq \tw+1$, and $a=|f^{-1}(1)|$, there are at
most $2^{a}$ possible pairs of correct colorings for $f$ (this follows from the definition of correct colorings). There are $2^{p-a} {p \choose a}$
possible colorings $f$ with $a$ vertices colored $1$, thus 
$$|\{(f_{1},f_{2}): f\in \{0,1,2\}^{p}, f_{1} \ and \ f_{2} \ are \ correct \ for \ f\}|\leq \displaystyle\sum_{a=0}^{p} 2^{p-a}{p \choose a}\cdot 2^{a}=4^{p}.$$

Since there are at most $4^{\tw}$ pairs of correct colorings, the evaluation of all join nodes altogether can be done in $4^{\tw}\cdot \tw^{\mathcal{O}(1)} \cdot n$ time. However, the fast subset convolution can be used to handle the join nodes more efficiently. In our case, set $B$ (given in Definition \ref{defsubsetconvo}) is a subset of a bag of $\mathcal {T}$. Further, we take the subset convolution over the max-sum semiring.

Note that if $f,f_{1},f_{2}:\mathcal{B}_{x} \rightarrow \{0,1,2\}$, then $f_{1}$ and $f_{2}$ are correct for a coloring $f$ if and only if the following conditions hold:

\begin{enumerate}
    \item [$C.1)$] $f^{-1}(0)=f_{1}^{-1}(0)=f_{2}^{-1}(0)$,
    \item [$C.2)$] $f^{-1}(1)=f_{1}^{-1}(1)\cup f_{2}^{-1}(1)$,
    \item [$C.3)$]  $f_{1}^{-1}(1)\cap f_{2}^{-1}(1)=\emptyset$.
\end{enumerate}

Another required condition, i.e., $f^{-1}(2)=f_{1}^{-1}(2)\cap f_{2}^{-1}(2)$ is already implied by conditions C.1)-C.3).  Next, note that if we fix $f^{-1}(0)$, then what we want to compute resembles the subset convolution. So, we fix a set $R \subseteq \mathcal{B}_{x}$. Further, let $\mathcal{F}_{R}$ denote the set of all functions $f:\mathcal{B}_{x}\rightarrow \{0,1,2\}$ such
that $f^{-1}(0)=R$. Next, we compute the values of $\mathcal{M}_{x}[f]$ for all $f\in \mathcal{F}_{R}$.
Note that every function $f\in \mathcal{F}_{R}$ can be represented by a set $S \subseteq \mathcal{B}_{x}\setminus R$,
namely, the preimage of 1. Hence, we can define the coloring represented by
$S$ as

\begin{equation} \label{eq1}
  g_{S}(x) = \begin{cases}
        0 \ \  \ \ \ \ \ \ \text{\ if \ $x\in R$,}
        \\
        1 \ \ \  \ \ \ \ \    \text{\ if \ $x\in S$,}
            \\
            2 \ \ \ \ \ \ \  \ \text{\ if \ $x\in \mathcal{B}_{x}\setminus (R\cup S)$}.
        \end{cases}
 \end{equation}

Now, for every $f\in \mathcal{F}_{R}$, we have
 \begin{equation}\label{eq2}
\mathcal{M}_{x}[f]=\displaystyle\max_{\substack{A\cup B=f^{-1}(1) \\ A\cap B =\emptyset}} \{\mathcal{M}_{y_{1}}[g_{A}]+\mathcal{M}_{y_{2}}[g_{B}]-|f^{-1}(1)|-|f^{-1}(2)|\}.
\end{equation}

The following observation follows from the definitions of $\mathcal{M}_{x}[f]$, $g_S(x)$, subset convolution, and equations (\ref{eq1}) and (\ref{eq2}).

\begin {observation} Let $v\in \{y_{1},y_{2}\}$ and $\mathcal{M}_{v}:2^{\mathcal{B}_{x}\setminus R}\rightarrow \{1,2,\ldots,k\}$  be such that for every
$S\subseteq \mathcal{B}_{x}\setminus R$, $\mathcal{M}_{v}(S)=\mathcal{M}_{v}[g_{S}]$. Then, for every $S\subseteq \mathcal{B}_{x} \setminus R$,

\begin{equation*}
\mathcal{M}_{x}[g_S]= (\mathcal{M}_{y_{1}}* \mathcal{M}_{y_{2}})(S)+|R|-|\mathcal{B}_{x}|,
\end{equation*}

where the subset convolution is over the max-sum semiring.
\end {observation}

By Proposition \ref{B2}, we compute $\mathcal{M}_{x}[g_{S}]$ for every $S\subseteq \mathcal{B}_{x} \setminus R$ in  $2^{|\mathcal{B}_{x}\setminus R|}\cdot \tw^{\mathcal{O}(1)}\cdot \mathcal{O}(k \log k \log \log k)=2^{|\mathcal{B}_{x}\setminus R|}\cdot \tw^{\mathcal{O}(1)}$ time. Also, we have to try all possible fixed subsets $R$ of $\mathcal{B}_{x}$. Since $\displaystyle\sum_{R\subseteq \mathcal{B}_{x}}2^{|\mathcal{B}_{x}\setminus R|}=3^{|\mathcal{B}_{x}|}\leq 3^{\tw+1}$, the total time spent for all subsets $R\subseteq \mathcal{B}_{x}$ is $3^{\tw}\cdot \tw^{\mathcal{O}(1)}$. So, clearly, the evaluation of all join nodes can be done in $3^{\tw} \cdot \tw^{\mathcal{O}(1)}\cdot n$ time.

 Thus from the description of all nodes, we have the following theorem.
\inducedmatching*

\smallskip

\section{Algorithm for Acyclic Matching} \label{AM}

In this section, we present a $6^{\tw} \cdot n^{\mathcal{O}(1)}$-time algorithm for \textsc{Acyclic Matching} assuming that we are given a nice tree decomposition
$\mathcal{T} = (\mathbb{T},\{\mathcal{B}_{x}\}_{x\in V(\mathbb{T})})$ of $G$ of width $\tw$ that satisfies the \textit{deferred edge} property. We use the Cut \& Count technique along with a concept called \textit{markers} (see \cite{cygan1}). 
Given that the $\textsc{Acyclic Matching}$ problem does not impose an explicit connectivity requirement, we can proceed by selecting the (presumed) forest obtained after choosing the vertices saturated by an acyclic matching $M$ and using the following result:
\begin{proposition} [\cite{cygan1}] \label{am1}
	 A graph $G$ with $n$ vertices and $m$ edges is a forest if and only if it has at most $n-m$ connected components.
\end{proposition}

Our solution set contains pairs $(X,P)$, where $X\subseteq V(G)$ is a set of $M$-saturated vertices
and $P\subseteq V(G)$ is a set of marked vertices (markers) such that each connected component in $G[X]$ contains at least one marked vertex. Markers will be helpful in bounding the number of connected components in $G[X]$ by $n'-m'$, where $n'=|X|$ and $m'$ is the number of edges in $G[X]$ (so that Proposition \ref{am1} can be applied). Since we will use the Isolation lemma (see Section \ref{PC}), we will be assigning random weights to the vertices of $X$. Furthermore, note that two pairs from our solution set with different sets of marked vertices are necessarily considered to be two different solutions. For this reason, we assign random weights both to the vertices of $X$ and vertices of $P$.

Throughout this section, as the universe, we take the set $U=V(G)\times\{\F,\K \}$, where $V(G)\times\{\F\}$ is used to assign weights to vertices of the chosen forest and $V(G)\times\{\K\}$ is used to assign weights to vertices chosen as markers. Also, throughout this section, we assume that we
are given a weight function $\mathsf{w}:U\rightarrow\{1,2,\ldots,N\}$, where $N=3|U|=6|V(G)|$.

Let us first consider the Cut part and start by defining the objects we are going to count.\\

\noindent \textbf{The Cut part:}

\begin{definition} \label{defcut}
 Let $G$ be a graph with $n$ vertices and $m$ edges. For integers $0\leq A\leq n, 0 \leq B \leq m, 0 \leq C \leq n,$ and $0 \leq W \leq 2Nn$, we define the following:
 
\begin{enumerate}
\item $\mathcal{R}_{W}^{A,B,C}=\{ (X,P): X\subseteq V(G) \wedge |X|=A \wedge G[X]$ contains
exactly $B$ edges $\wedge$ $G[X]$ has a perfect matching $\wedge$ $P\subseteq X$ $\wedge$ $|P|=C$ $\wedge$ $\mathsf{w}(X\times\{\F\})+\mathsf{w}(P\times\{\K\})=W\}$.
\item $\mathcal{S}_{W}^{A,B,C}=\{ (X,P)\in \mathcal{R}_{W}^{A,B,C}:$ $G[X]$ is a forest
containing at least one marker from the set $P$ in each connected component$\}$.
\item $\mathcal{C}_{W}^{A,B,C}=\{ ((X,P),(X_{l},X_{r})):(X,P)\in \mathcal{R}_{W}^{A,B,C}$ $\wedge$ $P\subseteq X_{l}$ $\wedge$ $(X_{l},X_{r})$ is a consistent cut of $G[X]\}$.
\end{enumerate}
\end{definition}

We call the set $\mathcal{R}=\bigcup_{A,B,C,W}\mathcal{R}_{W}^{A,B,C}$ the family of \textit{candidate solutions},  $\mathcal{S}=\bigcup_{A,B,C,W}\mathcal{S}_{W}^{A,B,C}$ the family of \textit{solutions}, and $\mathcal{C}=\bigcup_{A,B,C,W}\mathcal{C}_{W}^{A,B,C}$ the family of \textit{cuts}.\\

\noindent \textbf{The Count part:}

\begin{lemma}\label{equal}
	Let $G,A,B,C,W,\mathcal{C}_{W}^{A,B,C},$ and $\mathcal{S}_{W}^{A,B,C}$ be as defined in Definition \ref{defcut}. Then, for every $A,B,C,W$ satisfying $C\leq A-B$, we have $\big{|}\mathcal{C}_{W}^{A,B,C}\big{|}_{2}\equiv\big{|}\mathcal{S}_{W}^{A,B,C}\big{|}.$ 
\end{lemma}
\begin{proof}
Let $cc(P,G[X])$ denote the number of connected components of $G[X]$ that do not contain any marker from the set $P$. Observe that every connected component in $G[X]$ that does not contain any marker from the set $P$ has two options, it either lies entirely in $X_{l}$, or it lies entirely in $X_{r}$ (as $(X_{l},X_{r})$ is a consistent cut of $G[X]$). Since there are exactly $cc(P,G[X])$ such components, for any $A,B,C,W$, and $(X,P)\in \mathcal{R}_{W}^{A,B,C}$, there are exactly $2^{cc(P,G[X])}$ cuts $(X_l,X_r)$ such that $((X,P),(X_{l},X_{r}))\in \mathcal{C}_{W}^{A,B,C}$. Therefore, 
$$ |\mathcal{C}_{W}^{A,B,C}|_{2}=|\{(X,P)\in \mathcal{R}_{W}^{A,B,C}:cc(P,G[X])=0\}|.  $$

Next, note that if $cc(P,G[X])=0$, then every connected component in $G[X]$ has at least one marked vertex. This implies that the number of connected components in $G[X]$ is less than or equal to $C$. Therefore, for every $A,B,C,W$ satisfying $C\leq A-B$, by Proposition \ref{am1}, it is clear that $G[X]$ is a forest. Thus, we have $|\mathcal{C}_{W}^{A,B,C}|_{2}=|\mathcal{S}_{W}^{A,B,C}|$. \qed
\end{proof}

\begin{remark} Condition $C\leq A-B$ is necessary for Lemma \ref{equal} as otherwise (if $A-B<C$), it is not possible to bound the number of connected components in $G[X]$ by $A-B$. As a result, the elements of $\mathcal{S}_{W}^{A,B,C}$ could not be identified, and Proposition \ref{am1} could not be applied.
\end{remark}

By Isolation Lemma \cite{mulmuley}, we have the following lemma.
\begin{lemma} \label{iso1}
Let $G$ and $\mathcal{S}_{W}^{A,B,C}$ be as defined in Definition \ref{defcut}. For each $u\in U$, where $U$ is the universe, choose a weight $\mathsf{w}(u)\in \{1,2,\ldots, 3|U|\}$ uniformly and independently at random. For some $A,B,C,W$ satisfying $C\leq A-B$, if $ |\mathcal{S}_{W}^{A,B,C}|>0$, then 
$$ \mathsf{prob}\big( \mathsf{w} \ isolates \ \mathcal{S}_{W}^{A,B,C}\big)\geq  \frac{2}{3}.$$
\end{lemma}

The following observation helps us in proving Theorem \ref{am}.
\begin{observation} \label{obam}
$G$ admits an acyclic matching of size $\frac{\ell}{2}$ if and only if there exist integers $B$ and $W$ such that the set $\mathcal{S}_{W}^{\ell,B,\ell-B}$ is nonempty.
\end{observation}
\begin{proof}
In one direction, let $G$ admit an acyclic matching, say $M$, of size $\frac{\ell}{2}$. Let $X$ be the set of $M$-saturated vertices. By the definition of an acyclic matching, it is clear that $G[X]$ is a forest. Further, $|X|=\ell$. Let $B$ be the number of edges in $G[X]$. By Lemma \ref{iso1}, with non-zero probability, we can say that there exists a $W$ such that $\mathcal{S}_{W}^{\ell,B,\ell-B}$ is nonempty.

In the other direction, let $\mathcal{S}_{W}^{\ell,B,\ell-B}$ be nonempty. This implies that there exists a set, say $X\subseteq V(G)$, such that $|X|=\ell$ and the number of marked vertices in $X$ is $\ell-B$. Further, by the definition of $\mathcal{S}_{W}^{\ell,B,\ell-B}$, the number of connected components in $G[X]$ is bounded by $\ell-B$. By Proposition \ref{am1}, it implies that $G[X]$ is a forest. Further, since $G[X]$ has a perfect matching, it is an acyclic matching of size $\frac{\ell}{2}$ as required. \qed
\end{proof}

Now we describe a procedure that, given a nice tree decomposition $\mathcal{T}$ with the deferred edge property, a weight function $\mathsf{w}:U\rightarrow\{1,2,\ldots,N\}$, and integers $A,B,C,W$ as defined in Definition \ref{defcut} and satisfying $C\leq A-B$, computes $\big{|}\mathcal{C}_{W}^{A,B,C}\big{|}_{2}$  using dynamic programming. For this purpose, consider the following definition:

\begin{definition} \label{defR}

For every bag $\mathcal{B}_x$ of the tree
decomposition $\mathcal{T}$, for every integer $0\leq a \leq n, 0\leq b <n, 0\leq c \leq n, 0\leq w \leq 12n^{2}$, for every coloring $d: \mathcal{B}_{x} \rightarrow \{0,1,2\}$, for every coloring $s: \mathcal{B}_{x} \rightarrow \{0,l,r\}$, we define the following: 
\begin{enumerate}

\item $ \mathcal{R}_{x}[a,b,c,d,w]=\{(X,P): d \ is \ a \ valid \ coloring \ of \ \mathcal{B}_{x} \ \wedge \ X \ is \ the \ set \ of \ \\ vertices \ colored  \ 1 \ or \ 2 \ under \ some \ valid \ extension \ \widehat{d} \ of \ d \ in \ G_{x} \ \wedge \ |X|=a \ \wedge |E_{x}\cap E(G[X])|=b \ \wedge P\subseteq X\setminus \mathcal{B}_{x} \ \wedge |P|=c \ \wedge \mathsf{w}(X \times \{\F\})+\mathsf{w}(P \times \{\K\})= $ $w \}$.

\item $ \mathcal{C}_{x}[a,b,c,d,w]=\{((X,P),(X_{l},X_{r})): (X,P)\in \mathcal{R}_{x}[a,b,c,d,w] \ \wedge P\subseteq X_{l} \ \wedge (X,(X_{l},X_{r}))$ is a consistently cut subgraph of $G_{x}\}$.

\item $ \widetilde{\mathcal{A}}_{x}[a,b,c,d,w,s]=|\{((X,P),(X_{l},X_{r}))\in \mathcal{C}_{x}[a,b,c,d,w]: (s(v)=l\Rightarrow v\in X_{l}) \ \wedge (s(v)=r\Rightarrow v\in X_{r}) \ \wedge (s(v)=0\Rightarrow \ v\notin X)\}|$.

\end{enumerate}
\end{definition}

\begin{remark}\label{rmrk1}
  In Definition \ref{defR}, we assume $b<n$ because otherwise, an induced subgraph containing $b$ edges is definitely not a forest. 
\end{remark} 

The intuition behind Definition \ref{defR} is that the set $\mathcal{R}_{x}[a,b,c,d,w]$ contains all pairs $(X,P)$ that could potentially be extended to a candidate solution from $\mathcal{R}$ (with cardinality and weight restrictions as prescribed by $a,b,c,$ and $w$), and the set $\mathcal{C}_{x}[a,b,c,d,w]$ contains all consistently cut subgraphs of $G_{x}$ that could potentially be extended to elements of $\mathcal{C}$ (with cardinality and weight restrictions as prescribed by $a,b,c,$ and $w$). The number $\widetilde{\mathcal{A}}_{x}[a,b,c,d,w,s]$ counts precisely those elements of $\mathcal{C}_{x}[a,b,c,d,w]$ for which $s(v)$ describes whether for every $v\in \mathcal{B}_{x}$, $v$ lies in $X_{l}, X_{r},$ or outside $X$ depending on whether $s(v)$ is $l, r$, or $0$, respectively. 

We have a table $\mathcal{A}$ with an entry $\mathcal{A}_{x}[a,b,c,d,w,s]$ for each bag $\mathcal{B}_{x}$ of $\mathbb{T}$, for integers $0\leq a \leq n, 0\leq b <n, 0\leq c \leq n, 0\leq w \leq 12n^{2}$, for every coloring $d: \mathcal{B}_{x} \rightarrow \{0,1,2\}$, and for every coloring $s: \mathcal{B}_{x} \rightarrow \{0,l,r\}$. We say that $s$ and $d$ are compatible if for every $v\in \mathcal{B}_{x}$, the following hold:  $d(v)=0$ if and only if $s(v)=0$. Note that we have at most $\mathcal{O}(\tw\cdot n)$ many choices for $x$, at most $n$ choices for $a$, $b$, and $c$, at most $12n^{2}$ choices for $w$, and at most $5^{\tw}$  many compatible choices for $d$ and $s$. Whenever $s$ is not compatible with $d$, we do not store the entry $\mathcal{A}_{x}[a,b,c,d,w,s]$ and assume that the access to such an entry returns $0$. Therefore, the size of table $\mathcal{A}$ is bounded by $\mathcal{O}(5^{\tw}\cdot\tw \cdot n)$. We will show how to compute the table $\mathcal{A}$ so that the following will be satisfied.

\begin{lemma}\label{amimp}
If $d$ is valid and $d$ is compatible with $s$, then $\mathcal{A}_{x}[a,b,c,d,w,s]$ stores the value $\widetilde{\mathcal{A}}_{x}[a,b,c,d,w,s]$. Otherwise, the entry $\mathcal{A}_{x}[a,b,c,d,w,s]$ stores the value $0$.
\end{lemma}

By Lemma \ref{equal}, we are interested in values $\big{|}\mathcal{C}_{W}^{A,B,C}\big{|}_{2}$. By Observation \ref{obam}, Definition \ref{defR}, and Lemma \ref{amimp}, it suffices to compute values $ \mathcal{A}_{r}[\ell,B,\ell-B,\emptyset,W,\emptyset]$ for all $B$ and $W$, where $r$ is the root of the decomposition $\mathcal{T}$ (note that $ \mathcal{A}_{r}[\ell,B,\ell-B,\emptyset,W,\emptyset]=\big{|}\mathcal{C}_{W}^{\ell,B,\ell-B}\big{|}$, and we will calculate the modulo 2 separately). Further, to achieve the time complexity we aim to achieve, we have the following remark. 

\begin{remark}\label{rmrk2}
  We decide whether to mark a vertex or not in its forget bag. 
\end{remark}

Our algorithm computes $\mathcal{A}_{x}[a,b,c,d,w,s]$ for all bags $\mathcal{B}_x\in \mathcal{T}$ in a bottom-up manner for all integers $0\leq a \leq n, 0\leq b <n, 0\leq c \leq n, 0\leq w \leq 12n^{2}$, and for all compatible colorings $d: \mathcal{B}_{x} \rightarrow \{0,1,2\}$ and $s: \mathcal{B}_{x} \rightarrow \{0,l,r\}$. We now give the recurrences for $\mathcal{A}_{x}[a,b,c,d,w,s]$ that are used by our dynamic programming algorithm. \bigskip

\noindent \textbf{Leaf node:}  For a leaf node $x$, we have that $\mathcal{B}_{x}=\emptyset$. Hence there is only one possible coloring on $\mathcal{B}_{x}$, that is, the empty coloring. Since $d$ and $s$ are empty, the only compatible values of $a,b,c$, and $w$ are $0$, so we have  $$\mathcal{A}_{x}[0,0,0,\emptyset,0,\emptyset]=1.$$

All other values of $\mathcal{A}_{x}[a,b,c,d,w,s]$ are zero.
\bigskip

\noindent \textbf{Introduce vertex node:} Suppose that $x$ is an introduce vertex node with child node $y$ such that $\mathcal{B}_{x}=\mathcal{B}_{y}\cup \{v\}$ for some $v\notin \mathcal{B}_{y}$. 
 Note that we have not introduced any edges incident on $v$ so far, so $v$ is isolated in $G_x$. For every coloring $d: \mathcal{B}_{x}\rightarrow \{0,1,2\}$ and every coloring $s: \mathcal{B}_{x} \rightarrow \{0,l,r\}$ such that $d$ and $s$ are compatible, and for all integers $0\leq a \leq n, 0\leq b <n, 0\leq c \leq n, 0\leq w \leq 12n^{2}$, we have the following recursive formula:
 
 \begin{equation*}
  \mathcal{A}_{x}[a,b,c,d,w,s] = \begin{cases}
        \mathcal{A}_{y}[a,b,c,d|_{\mathcal{B}_{y}},w,s|_{\mathcal{B}_{y}}] \ \  \ \ \ \ \ \ \ \ \ \ \ \ \ \ \ \ \ \ \ \ \ \ \text{\ if \ $d(v)=0$,}
        \\
        0 \ \ \  \ \ \ \ \ \  \ \ \ \ \ \ \  \ \ \ \ \ \ \  \ \ \ \ \ \ \  \ \ \ \ \ \ \  \ \ \ \ \ \  \ \ \ \ \ \ \ \ \  \text{\ if \ $d(v)=1$,}
        \\
        
         \mathcal{A}_{y}[a-1,b,c,d|_{\mathcal{B}_{y}},w-\mathsf{w}((v,F)),s|_{\mathcal{B}_{y}}] \ \ \ \text{\ if \ $d(v)=2$}.
            
            \end{cases}
 \end{equation*}
 
 When $d(v)=0$ or $d(v)=2$, then $d$ is valid if and only if $d|_{\mathcal{B}_{y}}$ is valid, and $d$ and $s$ are compatible if and only if $d|_{\mathcal{B}_{y}}$ and $s|_{\mathcal{B}_{y}}$ are compatible. When $d(v)=0$, then $a,b$, and $w$ remain unaffected by the introduction of $v$. Further, by Remark \ref{rmrk2}, $c$ also remains unaffected. When $d(v)=2$, $b$ and $c$ remain unaffected by the introduction of $v$, but $a$ is reduced by $1$ and $w$ is reduced by $\mathsf{w}((u,F))$ in $\mathcal{B}_{y}$. When $d(v)=1$, then $d$ cannot be a valid coloring as $v$ does not have any neighbor in $G_{x}$ (by the definition of a valid coloring, $v$ needs a neighbor of color $1$ in $G_{x}$), and hence $\mathcal{A}_{x}[a,b,c,d,w,s]=0$.

 Clearly, the evaluation of all introduce vertex nodes can be done in $5^{\tw} \cdot n^{\mathcal{O}(1)}$ time.
\bigskip

\noindent \textbf{Introduce edge node:} 
Suppose that $x$ is an introduce edge node that introduces
an edge $uv$, and let $y$ be the child of $x$.  For every coloring $d: \mathcal{B}_{x}\rightarrow \{0,1,2\}$ and every coloring $s: \mathcal{B}_{x} \rightarrow \{0,l,r\}$ such that $d$ and $s$ are compatible, and for all integers $0\leq a \leq n, 0\leq b <n, 0\leq c \leq n, 0\leq w \leq 12n^{2}$, we consider the following cases:
\bigskip

 If at least one of $d(u)$ or  $d(v)$ is $0$, then
$$\mathcal{A}_{x}[a,b,c,d,w,s]=\mathcal{A}_{y}[a,b,c,d,w,s].$$

Else, if at least one of $d(u)$ or  $d(v)$ is $2$ and $s(u)=s(v)$, then
$$\mathcal{A}_{x}[a,b,c,d,w,s]=\mathcal{A}_{y}[a,b-1,c,d,w,s].$$

Else, if both $d(u)$ and $d(v)$ are $1$ and $s(u)=s(v)$, then
$$\mathcal{A}_{x}[a,b,c,d,w,s]=\displaystyle \max\{\mathcal{A}_{y}[a,b-1,c,d,w,s],\mathcal{A}_{y}[a,b-1,c,d_{\{u,v\}\rightarrow 2},w,s]\}.$$

All other values are zero.

Note that whenever an edge is introduced between $u$ and $v$ such that $d(u)$ and $d(v)$ are non-zero, then $s(u)=s(v)$.
 Clearly, the evaluation of all introduce edge nodes can be done in $5^{\tw}\cdot  n^{\mathcal{O}(1)}$ time.
\bigskip

\noindent  \textbf{Forget node:} Suppose that $x$ is a forget vertex node with
a child $y$ such that $\mathcal{B}_{x}=\mathcal{B}_{y}\setminus \{u\}$ for some $u\in \mathcal{B}_{y}$. For every coloring $d: \mathcal{B}_{x}\rightarrow \{0,1,2\}$ and every coloring $s: \mathcal{B}_{x} \rightarrow \{0,l,r\}$ such that $d$ and $s$ are compatible, and for all integers $0\leq a,c \leq n, 0\leq b <n, 0\leq w \leq 12n^{2}$, we have 

\begin{equation}\label{eqforget}
\begin{split}
\mathcal{A}_{x}[a,b,c,d,w,s] &  =  \mathcal{A}_{y}[a,b,c-1,d_{u \rightarrow 1},w-\mathsf{w}((u,P)),s[u\rightarrow l]]  
\\ &  + \mathcal{A}_{y}[a,b,c,d_{u \rightarrow 0},w,s[u\rightarrow 0]]   +
 \displaystyle\sum_{\alpha\in\{l,r\}}\mathcal{A}_{y}[a,b,c,d_{u \rightarrow 1},w,s[u\rightarrow \alpha]].\end{split}
\end{equation}

By Remark \ref{rmrk2}, we decide whether to mark a vertex or not in its forget bag. The first term on the right-hand side in (\ref{eqforget}) corresponds to the case when $d(u)=1$ in $\mathcal{B}_{y}$, and we decide to mark $u$. Since we mark the vertices only if it is colored $l$ under $s$ on $\mathcal{B}_{y}$, so $s(u)$ should be $l$ in $\mathcal{B}_{y}$. Further, $w$ is reduced by $\mathsf{w}((u,P))$ in $\mathcal{B}_{y}$. The second term on the right-hand side in (\ref{eqforget}) corresponds to the case when $d(u)=0$ in $\mathcal{B}_{y}$. The third term on the right-hand side in (\ref{eqforget}) corresponds to the case when $d(u)=1$ in $\mathcal{B}_{y}$, and we decide not to mark $u$. Note that in this case, it does not matter whether $s(u)=l$ or $s(u)=r$, and thus we take the summation over both the possibilities.

 Clearly, the evaluation of all forget nodes can be done in $5^{\tw}\cdot n^{\mathcal{O}(1)}$ time.
\bigskip

\noindent \textbf{Join node:}
Let $x$ be a join node with children $y_{1}$ and $y_{2}$. For every coloring $d: \mathcal{B}_{x}\rightarrow \{0,1,2\}$ and every coloring $s: \mathcal{B}_{x} \rightarrow \{0,l,r\}$ such that $d$ and $s$ are compatible, and for all integers $0\leq a,c \leq n, 0\leq b <n, 0\leq w \leq 12n^{2}$, we have

 \begin{equation}\label{joinam}
 \begin{split}
\mathcal{A}_{x}[a,b,c,d,w,s] & = \displaystyle\sum_{d_{1},d_{2}} \displaystyle\sum_{\substack{a_{1}+a_{2}=a\\ +|s^{-1}(\{l,r\})|}} \ \displaystyle\sum_{b_{1}+b_{2}=b} \ \displaystyle\sum_{c_{1}+c_{2}=c} \  \displaystyle\sum_{\substack{w_{1}+w_{2}=w\\+w(s^{-1}(\{l,r\})\times\{F\})}} 
  \\ & \ \ \mathcal{A}_{y_{1}}[a_{1},b_{1},c_{1},d_{1},w_{1},s]\cdot\mathcal{A}_{y_{2}}[a_{2},b_{2},c_{2},d_{2},w_{2},s],
  \end{split}
 \end{equation}

where $d_{1}: \mathcal{B}_{y_{1}}\rightarrow \{0,1,2\}$, $d_{2}: \mathcal{B}_{y_{2}}\rightarrow \{0,1,2\}$, $0\leq a_{1},a_{2},c_{1},c_{2} \leq n,$ $0\leq b_{1},b_{2} <n, 0\leq w_{1},w_{2} \leq 12n^{2}$ such that $s$ is compatible with $d_{1}$ and $d_{2}$, and $d_{1},d_{2}$ are correct for $d$.

The only valid combinations to achieve the coloring $s$ in (\ref{joinam}) is to have the same coloring in both children $y_{1}$ and $y_{2}$ as each vertex gets a unique color under $s$ in $G_{x}$. Since
vertices colored $l$ and $r$ in $\mathcal{B}_x$ are accounted for in both tables of the children, we add their contribution to the accumulators $a$ and $w$. Also, as no edges have been introduced yet among the vertices in $\mathcal{B}_{x}$, $b$ is equal to the sum of the values of $b_{1}$ and $b_{2}$. Also, by Remark \ref{rmrk2}, we can say that, $c=c_{1}+c_{2}$.

By the naive method, the evaluation for all join nodes altogether can be done in $7^{\tw} \cdot n^{\mathcal{O}(1)}$ time as follows. Note that if a pair $(d_{1},d_{2})$ is correct for $d$, and $s$ is compatible with $d$, then for every $v\in \mathcal{B}_{x}$, $(d(v),d_{1}(v),d_{2}(v),s(v))$ can have one of the following value
$$\{(0,0,0,0),(1,1,2,l),(1,1,2,r),(1,2,1,l),(1,2,1,r),(2,2,2,l),(2,2,2,r)\}.$$

It follows that there are exactly $7^{|\mathcal{B}_{x}|}$ tuples of colorings $(d,d_{1},d_{2},s)$ such that
$(d_{1},d_{2})$ are correct for $d$, and $s$ is compatible with $d$,$d_{1}$, and $d_{2}$, since for every vertex $v$ we have seven possibilities
for $(d(v),d_{1}(v),d_{2}(v),s(v))$. We iterate through all these tuples, and for each
tuple $(d(v),d_{1}(v),d_{2}(v),s(v))$, we include the contribution corresponding to $d_{1},d_{2}$ to the value of $\mathcal{A}_{x}(a,b,c,d,w,s)$
according to (\ref{joinam}). As $|\mathcal{B}_{x}|\leq \tw +1$, it follows that every join node takes $7^{\tw} \cdot n^{\mathcal{O}(1)}$ time. However, the fast subset convolution can be used to handle the join bags more efficiently. In our case, set $B$ (given in Definition \ref{defsubsetconvo}) is a bag of $\mathcal {T}$. To apply the fast subset convolution in (\ref{joinam}), first, let us define some notations.

For every $a\in [n]\cup \{0\}$ and $s: \mathcal{B}_{x} \rightarrow \{0,l,r\}$, let $(a_{y_{1}}^{1},a_{y_{2}}^{1}), (a_{y_{1}}^{2},a_{y_{2}}^{2}),\ldots,$ $(a_{y_{1}}^{k_{1}},a_{y_{2}}^{k_{1}})$ denote all possible pairs of integers such that $a_{y_{1}}^{i}+a_{y_{2}}^{i}=a
+|s^{-1}(\{l,r\})|$ for each $i\in[k_{1}]$. Note that for an $a$, we have $k_{1}=a+1$.

For every $b\in [n-1]\cup \{0\}$, let $(b_{y_{1}}^{1},b_{y_{2}}^{1}), (b_{y_{1}}^{2},b_{y_{2}}^{2}),\ldots, (b_{y_{1}}^{k_{2}},b_{y_{2}}^{k_{2}})$ denote all possible pairs of integers such that $b_{y_{1}}^{i}+b_{y_{2}}^{i}=b$ for each $i\in[k_{2}]$. Note that for a $b$, we have $k_{2}=b+1$. 

For every $c\in [n]\cup \{0\}$, let $(c_{y_{1}}^{1},c_{y_{2}}^{1}), (c_{y_{1}}^{2},c_{y_{2}}^{2}),\ldots, (c_{y_{1}}^{k_{3}},c_{y_{2}}^{k_{3}})$ denote all possible pairs of integers such that $c_{y_{1}}^{i}+c_{y_{2}}^{i}=c$ for each $i\in[k_{3}]$. Note that for a $c$, we have $k_{3}=c+1$. 

For every $w\in [12n^{2}]\cup \{0\}$ and $s: \mathcal{B}_{x} \rightarrow \{0,l,r\}$, let $(w_{y_{1}}^{1},w_{y_{2}}^{1}), \ldots,$ $(w_{y_{1}}^{k_{4}},w_{y_{2}}^{k_{4}})$ denote all possible pairs of integers such that $w_{y_{1}}^{i}+w_{y_{2}}^{i}=w+w(s^{-1}(\{l,r\})\times$ $\{F\})$ for each $i\in[k_{4}]$. Note that for a $w$, we have $k_{4}=w+1$.

\begin{remark}
Throughout this section, for fixed $a,c\in [n]\cup \{0\}$, $b\in [n-1]\cup \{0\}$, and $w\in [12n^{2}]\cup \{0\}$, $k_{1}=a+1, k_{2}=b+1, k_{3}=c+1,$ and $k_{4}=w+1$.
\end{remark}

For a join node $x$ with children $y_{1}$ and $y_{2}$, for fixed $a,c\in [n]\cup \{0\}$, $b\in [n-1]\cup \{0\}$, and $w\in [12n^{2}]\cup \{0\}$, for each $j\in [k_{1}]$, $l\in [k_{2}]$, $p\in [k_{3}]$, and $q\in [k_{4}]$, every coloring $d: \mathcal{B}_{x}\rightarrow \{0,1,2\}$, $d_{1}: \mathcal{B}_{y_{1}}\rightarrow \{0,1,2\}$, $d_{2}: \mathcal{B}_{y_{2}}\rightarrow \{0,1,2\}$, and $s: \mathcal{B}_{x} \rightarrow \{0,l,r\}$ such that $s$ is compatible with $d,d_{1},$ and $d_{2}$, let
\begin{equation} \label{joinam1}
\begin{split}
\mathcal{A}_{x}^{j,l,p,q}[a,b,c,d,w,s] =  \displaystyle\sum_{d_{1},d_{2}}
  \mathcal{A}_{y_{1}}[a_{y_{1}}^{j},b_{y_{1}}^{l},c_{y_{1}}^{p},d_{1},w_{y_{1}}^{q},s]\cdot\mathcal{A}_{y_{2}}[a_{y_{2}}^{j},b_{y_{2}}^{l},c_{y_{2}}^{p},d_{2},w_{y_{2}}^{q},s].
  \end{split}
  \end{equation}

By (\ref{joinam}) and (\ref{joinam1}), we have
\begin{equation} \label{joinam2}
\mathcal{A}_{x}[a,b,c,d,w,s]= \displaystyle\sum_{d_{1},d_{2}}\displaystyle\sum_{j=1}^{a+1} \ \displaystyle\sum_{l=1}^{b+1}\ \displaystyle\sum_{p=1}^{c+1}\ 
 \displaystyle\sum_{q=1}^{w+1}  \mathcal{A}_{x}^{j,l,p,q}[a,b,c,d,w,s].
\end{equation}

Let us compute (\ref{joinam1}) using the fast subset convolution. Note that if $d,d_{1},d_{2}:\mathcal{B}_{x} \rightarrow \{0,1,2\}$, then $d_{1}$ and $d_{2}$ are correct for a coloring $d$ if and only if the following conditions hold:

\begin{enumerate}
    \item [D.1)] $d^{-1}(0)=d_{1}^{-1}(0)=d_{2}^{-1}(0)$,
    \item [D.2)] $d^{-1}(1)=d_{1}^{-1}(1)\cup d_{2}^{-1}(1)$,
    \item [D.3)]  $d_{1}^{-1}(1)\cap d_{2}^{-1}(1)=\emptyset$.
\end{enumerate}

The condition $d^{-1}(2)=d_{1}^{-1}(2)\cap d_{2}^{-1}(2)$ is already implied by conditions D.1)-D.3).  Next, note that if we fix $d^{-1}(0)$, then what we want to compute resembles the subset convolution. So, let us fix a set $R \subseteq \mathcal{B}_{x}$. Further, let $\mathcal{F}_{R}$ denote the set of all functions $d:\mathcal{B}_{x}\rightarrow \{0,1,2\}$ such
that $d^{-1}(0)=R$. Next, we compute the values of $\mathcal{A}_{x}(a,b,c,d,w,s)$ for all $d\in \mathcal{F}_{R}$.
Note that every function $d\in \mathcal{F}_{R}$ can be represented by a set $S \in \mathcal{B}_{x}\setminus R$,
namely, the preimage of 1. Hence, we can define the coloring represented by
$S$ as
\begin{equation*}
  g_{S}(x) = \begin{cases}
        0 \ \  \ \ \ \ \ \ \text{\ if \ $x\in R$,}
        \\
        1 \ \ \  \ \ \ \ \    \text{\ if \ $x\in S$,}
            \\
            2 \ \ \ \ \ \ \  \ \text{\ if \ $x\in \mathcal{B}_{x}\setminus (R\cup S)$}.
        \end{cases}
 \end{equation*}
 
Now, for fixed $j\in [k_{1}]$, $l\in [k_{2}]$, $p\in [k_{3}]$, and $q\in [k_{4}]$ and for every $d\in \mathcal{F}_{R}$, (\ref{joinam1}) can be written as

 \begin{equation*}
 \mathcal{A}^{j,l,p,q}_{x}[a,b,c,d,w,s]= \displaystyle\sum_{\substack{A\cup B=S \\ A\cap B =\emptyset}} 
 \mathcal{A}_{y_{1}}[a_{y_{1}}^{j},b_{y_{1}}^{l},c_{y_{1}}^{p},g_{A},w_{y_{1}}^{q},s]\cdot\mathcal{A}_{y_{2}}[a_{y_{2}}^{j},b_{y_{2}}^{l},c_{y_{2}}^{p},g_{B},w_{y_{2}}^{q},s].
\end{equation*}

\begin{observation}Let $v\in \{y_1,y_2\}$, $s: \mathcal{B}_{x} \rightarrow \{0,l,r\}$ be a fixed coloring, $j\in [k_{1}]$, $l\in [k_{2}]$, $p\in [k_{3}]$, and $q\in [k_{4}]$ be fixed integers, and $\mathcal{A}_{{i},g_{S}}^{j,l,p,q,s}:2^{\mathcal{B}_{x}\setminus R}\rightarrow \{1,2,\ldots,k\}$ be such that for every
$S\subseteq \mathcal{B}_{x}\setminus R$, $\mathcal{A}_{{v}}^{j,l,p,q,s}(S)=\mathcal{A}_{v}(a_{v}^{j},b_{v}^{l},c_{v}^{p},g_{S},w_{v}^{q},s)$. Then, for every $S\subseteq \mathcal{B}_{x} \setminus R$,
\begin{equation*}
\mathcal{A}_{x}^{j,l,p,q}[a,b,c,g_{S},w,s]=   (\mathcal{A}_{y_1}^{j,l,p,q,s}* \mathcal{A}_{y_2}^{j,l,p,q,s}) (S),
\end{equation*}

where the subset convolution has its usual meaning, i.e., the sum of products.
\end{observation}

By Proposition \ref{B2}, we can compute $\mathcal{A}_{x}^{j,l,p,q}(a,b,c,g_{S},w,s)$ for every $S\subseteq \mathcal{B}_{x} \setminus R$ in  $2^{|\mathcal{B}_{x}\setminus R|}\cdot \tw^{\mathcal{O}(1)}\cdot \mathcal{O}(k \log k \log \log k)=2^{|\mathcal{B}_{x}\setminus R|}\cdot \tw^{\mathcal{O}(1)}$ time. Therefore, for each $j\in [k_{1}]$, $l\in [k_{2}]$, $p\in [k_{3}]$, and $q\in [k_{4}]$, we can compute $\mathcal{A}_{x}^{j,l,p,q}[a,b,c,g_{S},w,s]$ for every $S \subseteq \mathcal{B}_{x}\setminus R$ (or for every $d\in \mathcal{F}_{R}$) in $2^{|\mathcal{B}_{x}\setminus R|}\cdot n^{5} \cdot \tw^{\mathcal{O}(1)}$ time. In the same time, we can compute $\mathcal{A}_{x}[a,b,c,d,w,s]$ by (\ref{joinam2}) for every $d\in \mathcal{F}_{R}$.  Also, we have to try all possible fixed subsets $R$ of $\mathcal{B}_{x}$. Since $\displaystyle\sum_{R\subseteq \mathcal{B}_{x}}2^{|\mathcal{B}_{x}\setminus R|}=3^{|\mathcal{B}_{x}|}\leq 3^{\tw+1}$, the total time spent for all subsets $R\subseteq \mathcal{B}_{x}$ is $3^{\tw} \cdot n^{5} \cdot \tw^{\mathcal{O}(1)}$. 

Since for every $d$, there are at most $2^{\tw}$ choices for compatible $s$, the evaluation of all join nodes can be done in $6^{\tw} \cdot n^{\mathcal{O}(1)}$ time.

By Observation \ref{obam}, we have the following theorem,
\acyclicmatching*

\section{Algorithm for $c$-Disconnected Matching} \label{CDM}

In this section, we present a $(3c)^{\tw}\cdot  \tw^{\mathcal{O}(1)}\cdot n$-time algorithm for $c$-\textsc{Disconnected Matching} assuming that we are given a nice tree decomposition
$\mathcal{T} = (\mathbb{T},\{\mathcal{B}_{x}\}_{x\in V(\mathbb{T})})$ of $G$ of width $\tw$ that satisfies the \textit{deferred edge} property. For this purpose, we define the following notion.

\begin{definition} [Fine Coloring] \label{deffine}
Given a node $x$ of $\mathbb{T}$ and a fixed integer $c \geq 2$, a coloring $f:\mathcal{B}_{x}\rightarrow\{0,1,\ldots,c\}$ is a \emph{fine coloring} on $\mathcal{B}_{x}$ if there exists a coloring $\widehat{f}:V_{x}\rightarrow\{0,1,\ldots,c\}$ in $G_{x}$, called a \emph{fine extension} of $f$, such that the following hold: 
\begin{enumerate}
\item [$(i)$] $\widehat{f}$ restricted to $\mathcal{B}_{x}$ is exactly $f$.
\item  [$(ii)$] If $uv \in E_{x}$, $\widehat{f}(u)\neq 0$ and $\widehat{f}(v)\neq 0$, then $\widehat{f}(u)= \widehat{f}(v)$.
\end{enumerate}
\end{definition}

Note that point (ii) in Definition \ref{deffine} implies that whenever two vertices in $G_{x}$ have an edge between them, then they should get the same color under a fine extension except possibly when either of them is colored $0$.

Before we begin the formal description of the algorithm, let us briefly discuss the idea that yields us a single exponential running time for the \textsc{$c$-Disconnected Matching} problem rather than a slightly exponential running time\footnote{That is, running time $2^{\mathcal{O}(\tw)}\cdot n^{\mathcal{O}(1)}$ rather than $\tw^{\mathcal{O}(\tw)}\cdot n^{\mathcal{O}(1)}.$} (which is common for most of the naive dynamic programming algorithms for connectivity type problems). We will use Definition \ref{deffine} to partition the vertices of $V_{x}$ into color classes (at most $c$). Note that we do not require in Definition \ref{deffine} that $G_{x}[\widehat{f}^{-1}(i)]$ for any $i\in [c]$ is a connected graph. This is the crux of our efficiency. Specifically, this means that we do not keep track of the precise connected components of $G[V_{M}]$ in $G_{x}$ for a matching $M$, yet Definition \ref{deffine} is sufficient for us. 

Now, let us discuss our ideas more formally. We have a table $\mathcal{A}$ with an entry $\mathcal{A}_{x}[d,f,\widehat{c}]$ for each bag $\mathcal{B}_x$, for every coloring $d:\mathcal{B}_{x} \rightarrow \{0,1,2\}$, for every coloring $f:\mathcal{B}_{x} \rightarrow \{0,1,\ldots,c\}$, and for every set $\widehat{c}\subseteq \{1,2,\ldots,c\}$.  We say that $d$ and $f$ are \emph{compatible} if for every $v\in \mathcal{B}_{x}$, the following hold:  $d(v)=0$ if and only if $f(v)=0$. We say that $f$ and $\widehat{c}$ are compatible if for any $v\in \mathcal{B}_{x}$, $f(v)\in \widehat{c}$. Note that we have at most $\mathcal{O}(\tw\cdot n)$ many choices for $x$, at most $(2c+1)^{\tw}$ many choices for compatible $d$ and $f$, and at most $2^{c}$ choices for $\widehat{c}$. Furthermore, whenever $f$ is not compatible with $d$ or $\widehat{c}$, we do not store the entry $\mathcal{A}_{x}[d,f,\widehat{c}]$ and assume that the access to such an entry returns $-\infty$. Therefore, the size of table $\mathcal{A}$ is bounded by $\mathcal{O}((2c+1)^{\tw}\cdot \tw\cdot 
n)$. The following definition specifies the value each entry $\mathcal{A}_{x}[d,f,\widehat{c}]$ of $\mathcal{A}$ is supposed to store.

\begin{definition}
If $d$ is valid, $f$ is fine, $f$ is compatible with $d$ and $\widehat{c}$, and there exists a fine extension $\widehat{f}$ of $f$ such that $\widehat{c}$ equals the set of distinct non-zero colors assigned by $\widehat{f}$, then the entry $\mathcal{A}_{x}[d,f,\widehat{c}]$ stores the maximum number of vertices that are colored $1$ or $2$ under some valid extension $\widehat{d}$ of $d$ in $G_{x}$ such that for every $v\in V_{x}$, $\widehat{d}(v)=0$ if and only if $\widehat{f}(v)=0$. Otherwise, the entry $\mathcal{A}_{x}[d,f,\widehat{c}]$ stores the value $-\infty$.
\end{definition}

 See Figure \ref{fig1} for an illustration of how to compute $\widehat{c}$, given a valid extension $\widehat{d}$ of $d$ and a fine extension $\widehat{f}$ of $f$. 

\begin{figure}[t]
    \centering
    \includegraphics[scale=0.9]{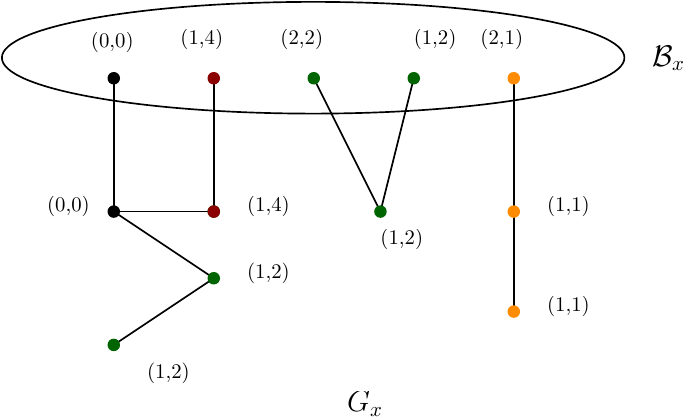}
    \caption{ The entry written beside each vertex is its $(\widehat{d},\widehat{f})$ value. Here, note that $\widehat{c}=\{1,2,4\}$ as $\widehat{f}$ uses these colors on $G_{x}$.}
    \label{fig1}
\end{figure}

Since the root of $\mathbb{T}$ is an empty node, note that the maximum number of vertices saturated by any $c$-disconnected matching is exactly $\mathcal{A}_{r}[\emptyset,\emptyset,\{1,2,\ldots,c\}]$, where $r$ is the root of $\mathbb{T}$.

We now provide recursive formulas to compute the entries of table $\mathcal{A}$.

\bigskip

\noindent \textbf{Leaf node:} For a leaf node $x$, we have that $\mathcal{B}_{x}=\emptyset$. Hence there is only one possible coloring on $\mathcal{B}_{x}$, that is, the empty coloring (for both $d$ and $f$). Since $f$ and $G_{x}$ are empty, the only compatible choice for $\widehat{c}$ is $\{\}$, and we have $\mathcal{A}_{x}[\emptyset,\emptyset,\{\}]=0$.

\bigskip
    
\noindent \textbf{Introduce vertex node:} Suppose that $x$ is an introduce vertex node with child node $y$ such that $\mathcal{B}_{x}=\mathcal{B}_{y}\cup \{v\}$ for some $v\notin \mathcal{B}_{y}$. Note that we have not introduced any edges incident on $v$ so far, so $v$ is isolated in $G_x$. For every coloring $d: \mathcal{B}_{x}\rightarrow \{0,1,2\}$, every set $\widehat{c} \subseteq \{1,2,\ldots,c\}$, and every coloring $f:\mathcal{B}_{x}\rightarrow\{0,1,\ldots,c\}$ such that $f$ is compatible with $d$ and $\widehat{c}$, we have the following recursive formula:

 \begin{equation}\label{intronode}
  \mathcal{A}_{x}[d,f,\widehat{c}] = \begin{cases}
        \mathcal{A}_{y}[d|_{\mathcal{B}_{y}},f|_{\mathcal{B}_{y}},\widehat{c}] \ \  \ \ \ \ \ \  \ \ \ \ \ \ \ \ \  \ \ \  \ \ \ \ \ \ \ \ \ \ \ \ \ \ \ \text{\ if \ $d(v)=0$,}
        \\
        -\infty \ \ \  \ \ \ \ \ \  \ \ \ \ \ \ \ \ \ \ \ \ \ \ \  \  \ \ \ \ \ \ \  \ \ \ \ \ \ \ \ \  \ \ \ \ \ \ \ \ \  \text{\ if \ $d(v)=1$,}
        \\
        
         \displaystyle\max\{ \mathcal{A}_{y}[d|_{\mathcal{B}_{y}},f|_{\mathcal{B}_{y}},\widehat{c}\setminus \{f(v)\}]+1, \\ \mathcal{A}_{y}[d|_{\mathcal{B}_{y}},f|_{\mathcal{B}_{y}},\widehat{c}]+1\} \   \  \ \ \ \ \ \ \   \ \ \ \ \  \ \ \ \ \  \ \ \ \ \   \  \ \  \ \text{\ if \ $d(v)=2$}.
            
            \end{cases}
 \end{equation}

When $d(v)=0$ or $d(v)=2$, then $d$ is valid if and only if $d|_{\mathcal{B}_{y}}$ is valid, and $f$ is fine if and only if $f|_{\mathcal{B}_{y}}$ is fine. Further, when $d(v)=0$, $\widehat{c}$ remains unaffected by the introduction of $v$. So, the case when $d(v)=0$ is correct in (\ref{intronode}). When $d(v)=2$, then there are two possibilities. One is that $v$ has been assigned a color under $f$ that has not been assigned to any other vertex in $V_{x}$, and the second is that $v$ under $f$ shares its color with some vertex ($\neq v$) in $V_{x}$, i.e., $f(v)\in \widehat{c}$.
For the first case, $f$ is compatible with $\widehat{c}$ if and only if $f|_{\mathcal{B}_{y}}$ is compatible with $\widehat{c}\setminus \{f(v)\}$. For the second case, $f$ is compatible with $\widehat{c}$ if and only if $f|_{\mathcal{B}_{y}}$ is compatible with $\widehat{c}$. Further, we increment the value by one in both cases as one more vertex has been colored $2$ in $G_{x}$. Thus, the case when $d(v)=2$ is correct in (\ref{intronode}). When $d(v)=1$, then $d$ cannot be a valid coloring as $v$ does not have any neighbor in $G_{x}$ (by the definition of a valid coloring, $v$ needs a neighbor of color $1$ in $G_{x}$), and hence $\mathcal{A}_{x}[d,f,\widehat{c}]=-\infty$. Thus, the case when $d(v)=1$ is correct in (\ref{intronode}).

Clearly, the evaluation of all introduce vertex nodes can be done in $(2c+1)^{\tw}\cdot \tw^{\mathcal{O}(1)} \cdot n$ time.
\bigskip

\noindent \textbf{Introduce edge node:} 
Suppose that $x$ is an introduce edge node that introduces
an edge $uv$, and let $y$ be the child of $x$. For every coloring $d: \mathcal{B}_{x}\rightarrow \{0,1,2\}$, every set $\widehat{c} \subseteq \{1,2,\ldots,c\}$, and every coloring $f:\mathcal{B}_{x}\rightarrow\{0,1,\ldots,c\}$ such that $f$ is compatible with $d$ and $\widehat{c}$, we consider the following cases:
\medskip

 If at least one of $d(u)$ or  $d(v)$ is $0$, then
\begin{equation*}\mathcal{A}_{x}[d,f,\widehat{c}]=\mathcal{A}_{y}[d,f,\widehat{c}].\end{equation*}

Else, if at least one of $d(u)$ or  $d(v)$ is $2$, and $f(u)=f(v)$, then
\begin{equation*}\mathcal{A}_{x}[d,f,\widehat{c}]=\mathcal{A}_{y}[d,f,\widehat{c}].\end{equation*}

Else, if both $d(u)$ and $d(v)$ are $1$, and $f(u)=f(v)$, then
\begin{equation} \label{new} \mathcal{A}_{x}[d,f,\widehat{c}]=\displaystyle \max\{\mathcal{A}_{y}[d,f,\widehat{c}],\mathcal{A}_{y}[d_{\{u,v\}\rightarrow 2},f,\widehat{c}]\}.\end{equation}

Else, $\mathcal{A}_{x}[d,f,\widehat{c}]=-\infty$.
\medskip

Note that if either $d(u)$ or $d(v)$ is $0$, then $d$ is valid on $\mathcal{B}_{x}$ if and only if $d$ is valid on $\mathcal{B}_{y}$ (by the definition of a valid coloring), and $f$ is fine on $\mathcal{B}_{x}$ if and only if $f$ is fine on $\mathcal{B}_{y}$ (by the definition of a fine coloring). Furthermore, as no new colors have been introduced under $f$, $\widehat{c}$ remains unaffected by the introduction of edge $uv$. The same arguments hold when either $d(u)$ or $d(v)$ is $2$. Next, whenever an edge is introduced between $u$ and $v$ such that $d(u)$ and $d(v)$ are non-zero, then $f$ is a fine coloring if and only if $f(u)=f(v)$ (by (ii) in Definition \ref{deffine}).  Next, let us consider the case when both $d(u)$ and $d(v)$ are $1$. In this case, there are two possibilities. One is where both $u$ and $v$ have already been matched to some other vertices in $G_{y}$, i.e., not with each other (corresponds to the first term on the right-hand side in (\ref{new})), and the other is when both $u$ and $v$ have not found their $M$-mate in $G_{y}$ and are matched to each other in $\mathcal{B}_{x}$ (corresponds to the second term on the right-hand side in (\ref{new})). We take the maximum over the two possibilities, and thus (\ref{new}) is correct.

Clearly, the evaluation of all introduce edge nodes can be done in $(2c+1)^{\tw}\cdot \tw^{\mathcal{O}(1)} \cdot n$ time.
 
\bigskip
 
 \noindent  \textbf{Forget node:} Suppose that $x$ is a forget vertex node with
a child $y$ such that $\mathcal{B}_{x}=\mathcal{B}_{y}\setminus \{u\}$ for some $u\in \mathcal{B}_{y}$. For every coloring $d: \mathcal{B}_{x}\rightarrow \{0,1,2\}$, every set $\widehat{c} \subseteq \{1,2,\ldots,c\}$, and every coloring $f:\mathcal{B}_{x}\rightarrow\{0,1,\ldots,c\}$ such that $f$ is compatible with $d$ and $\widehat{c}$, we have
 \begin{equation}\label{eq13}\mathcal{A}_{x}[d,f,\widehat{c}]=
\displaystyle\max\{\mathcal{A}_{y}[d_{u \rightarrow 0},f_{u \rightarrow 0},\widehat{c}],\displaystyle\max_{\overline{c}\in \widehat{c}}\{\mathcal{A}_{y}[d_{u \rightarrow 1},f_{u \rightarrow \overline{c}},\widehat{c}]\}\}.
\end{equation}

 The first term on the right-hand side in (\ref{eq13}) corresponds to the case when $d(u)=0$ in $\mathcal{B}_{y}$. Since we store only compatible values of $f$ and $d$, $f(u)=0$ in $\mathcal{B}_{y}$. The value of $\widehat{c}$ does not change as we only forget a vertex, or, in other words, $G_{x}=G_{y}$. The second term on the right-hand side in (\ref{eq13}) corresponds to the case when $d(u)=1$ in $\mathcal{B}_{y}$. In this case, we take the maximum over all possible values of $f(u)$. The only compatible choices for $f(u)$ are from the set $\widehat{c}$. Furthermore, note that the (outer) maximum is taken over colorings $d_{u \rightarrow 0}$  and $d_{u \rightarrow 1}$  only, as the coloring $d_{u \rightarrow 2}$ cannot be extended to a valid coloring once $u$ is forgotten (this follows by the definition of a valid coloring).
 
 Clearly, the evaluation of all forget nodes can be done in $(2c+1)^{\tw}\cdot \tw^{\mathcal{O}(1)} \cdot n$ time.
 
\bigskip
 
\noindent \textbf{Join node:}
Let $x$ be a join node with children $y_{1}$ and $y_{2}$. For every coloring $d: \mathcal{B}_{x}\rightarrow \{0,1,2\}$, every set $\widehat{c} \subseteq \{1,2,\ldots,c\}$, and for every coloring $f:\mathcal{B}_{x}\rightarrow\{0,1,\ldots,c\}$ such that $f$ is compatible with $d$ and $\widehat{c}$,  we have 
\begin{equation} \label{eqjoin}
\mathcal{A}_{x}[d,f,\widehat{c}]=\displaystyle\max_{d_{1},d_{2}}\{ \displaystyle\max_{\substack{\widehat{c}_{y_1},\widehat{c}_{y_2}\\ \widehat{c}_{y_1}\cup \widehat{c}_{y_2}= \widehat{c}}} \{\mathcal{A}_{y_{1}}[d_{1},f,\widehat{c}_{y_1}]+\mathcal{A}_{y_{2}}[d_{2},f,\widehat{c}_{y_2}]-|d^{-1}(1)|-|d^{-1}(2)|\}\},
\end{equation}

where $d_{1}: \mathcal{B}_{y_{1}}\rightarrow \{0,1,2\}$, $d_{2}: \mathcal{B}_{y_{2}}\rightarrow \{0,1,2\}$, $\widehat{c}_{y_1}, \widehat{c}_{y_2} \subseteq \{1,2,\ldots,c\}$ such that $f$ is compatible with $d_{1}$, $d_{2}$, $\widehat{c}_{y_1}$, and $\widehat{c}_{y_2}$, and $d_{1},d_{2}$ are correct for $d$. 

Note that we are determining the value of $\mathcal{A}_{x}[d,f,\widehat{c}]$ by looking up the corresponding values in nodes $y_{1}$ and $y_{2}$, adding the corresponding values, and subtracting the number of vertices
 colored $1$ or $2$ under $d$ (else, by Observation \ref{joinobs}, the number of vertices colored 1 or 2 in $\mathcal{B}_{x}$ would be counted twice). Note that for nodes $y_{1}$ and $y_{2}$, we are considering only those colorings of $\mathcal{B}_{y_{1}}$ and $\mathcal{B}_{y_{2}}$ that are correct for $d$, otherwise $d$ will not be valid.

By the naive method, the evaluation of all join nodes altogether can be done in $(4c+2)^{\tw} \cdot \tw^{\mathcal{O}(1)} \cdot n$ time as there are at most $2^{\tw}$ possible values of $(d_{1},d_{2})$ that are correct for $d$. However, the fast subset convolution can be used to handle join nodes more efficiently. In our case, set $B$ (given in Definition \ref{defsubsetconvo}) is a subset of a bag of $\mathcal {T}$. Further, we take the subset convolution over the max-sum semiring. To apply the fast subset convolution in (\ref{eqjoin}), first, let us define some notations.

For every $\widehat{c}\subseteq \{1,2,\ldots,c\}$, we define a set $\widehat{C}$ that contains all possible pairs $(\widehat{c}_{y_{1}}^{1},\widehat{c}_{y_{2}}^{1}),$ 
$(\widehat{c}_{y_{1}}^{2},\widehat{c}_{y_{2}}^{2}),\ldots,$ $ (\widehat{c}_{y_{1}}^{k_{1}},\widehat{c}_{y_{2}}^{k_{1}})$ such that $\widehat{c}_{y_{1}}^{i}\cup \widehat{c}_{y_{2}}^{i}=\widehat{c}$ for each $i\in[k_{1}]$. Note that for a $\widehat{c}$, we have $k_{1}={3}^{|\widehat{c}|}$ (depending on whether an element of $\widehat{c}$ lies in $\widehat{c}_{y_{1}}^{i}$ only, $\widehat{c}_{y_{2}}^{i}$ only, or both). Further, for every fixed $f:\mathcal{B}_{x}\rightarrow\{0,1,\ldots,c\}$, define a set $\widehat{C}_{f}$ that contains only those elements from $\widehat{C}$ that are compatible with $f$.
\medskip 

\begin{remark}
Throughout this section, for fixed $\widehat{c}$ and $f$, we set $k_{2}=|\widehat{C}_{f}|$. Note that, as $k_{1}={3}^{|\widehat{c}|}$, we have $k_{2}\leq {3}^{|\widehat{c}|}.$
\end{remark}

 For a join node $x$ with children $y_{1}$ and $y_{2}$, for every fixed coloring $d: \mathcal{B}_{x}\rightarrow \{0,1,2\}$, for every fixed coloring $f:\mathcal{B}_{x}\rightarrow\{0,1,\ldots,c\}$, and for every fixed $\widehat{c} \subseteq \{1,2,\ldots,c\}$ such that $f$ is compatible with $d$ and $\widehat{c}$, let 
\begin{equation}\label{eqjoin1}
\mathcal{A}_{x}^{i}[d,f,\widehat{c}]=\displaystyle\max_{d_{1},d_{2}} \{\mathcal{A}_{y_{1}}[d_{1},f,\widehat{c}_{y_{1}}^{i}]+\mathcal{A}_{y_{2}}[d_{2},f,\widehat{c}_{y_{2}}^{i}]-|d^{-1}(1)|-|d^{-1}(2)|\},
\end{equation}
where $d_{1}: \mathcal{B}_{y_{1}}\rightarrow \{0,1,2\}$, $d_{2}: \mathcal{B}_{y_{2}}\rightarrow \{0,1,2\}$ such that $d_{1},d_{2}$ are correct for $d$, and $f$ is compatible with $d_{1}$ and $d_{2}$, and $(\widehat{c}_{y_1}^{i}, \widehat{c}_{y_2}^{i}) \in \widehat{C}_{f}$.  

By (\ref{eqjoin}) and (\ref{eqjoin1}), we have 
\begin{equation} \label{eq5}
  \mathcal{A}_{x}[d,f,\widehat{c}]=\displaystyle\max_{1\leq i \leq k_{2}}\mathcal{A}_{x}^{i}[d,f,\widehat{c}].
\end{equation}

From the description of join nodes in Section \ref{IM}, recall that if $d,d_{1},d_{2}:\mathcal{B}_{x} \rightarrow \{0,1,2\}$, then $d_{1}$ and $d_{2}$ are correct for a coloring $d$ if and only if C.1)-C.3) hold.

  Next, note that if we fix $d^{-1}(0)$, then what we want to compute resembles the subset convolution. So, we fix a set $R \subseteq \mathcal{B}_{x}$. Further, let $\mathcal{F}_{R}$ denote the set of all functions $d:\mathcal{B}_{x}\rightarrow \{0,1,2\}$ such
that $d^{-1}(0)=R$. Next, for all $d\in \mathcal{F}_{R}$, we compute the values of $\mathcal{A}_{x}^{i}[d,f,\widehat{c}]$ for each $i\in [k_{2}]$.
Note that every function $d\in \mathcal{F}_{R}$ can be represented by a set $S \subseteq \mathcal{B}_{x}\setminus R$,
namely, the preimage of 1. Hence, we can define the coloring represented by
$S$ as in (\ref{eq1}).

Now, for every fixed $f:\mathcal{B}_{x}\rightarrow\{0,1,\ldots,c\}$, for every fixed $i\in [k_{2}]$, and for every $d\in \mathcal{F}_{R}$, (\ref{eqjoin1}) can be written as
 \begin{equation}\label{eq4}
\mathcal{A}_{x}^{i}[d,f,\widehat{c}]=\displaystyle\max_{\substack{A\cup B=d^{-1}(1) \\ A\cap B =\emptyset}} \{\mathcal{A}_{y_{1}}[g_{A},f,\widehat{c}_{y_{1}}^{i}]+\mathcal{A}_{y_{2}}[g_{B},f,\widehat{c}_{y_{2}}^{i}]-|d^{-1}(1)|-|d^{-1}(2)|\}.
\end{equation}

The following observation follows from the definitions of $g_S(x)$, subset convolution, $\mathcal{A}_{x}^{i}[d,f,\widehat{c}]$ for each $i\in [k_{2}]$, and equations (\ref{eq1}) and (\ref{eq4}).

\begin {observation} Let $v\in \{y_{1},y_{2}\}$, $f:\mathcal{B}_{x}\rightarrow\{0,1,\ldots,c\}$ be a fixed coloring, $i \in [k_{2}]$ be a fixed integer, and $\mathcal{A}_{v}^{f,i}:2^{\mathcal{B}_{x}\setminus R}\rightarrow \{1,2,\ldots,k\}$ be such that for every
$S\subseteq \mathcal{B}_{x}\setminus R$, $\mathcal{A}_{v}^{f,i}(S)=\mathcal{A}_{v}[g_{S},f,\widehat{c}_{v}^{i}]$. Then, for every $S\subseteq \mathcal{B}_{x} \setminus R$,
\begin{equation*}
\mathcal{A}_{x}^{i}[g_S,f,\widehat{c}]= (A_{y_{1}}^{f,i}* A_{y_{2}}^{f,i})(S)+|R|-|\mathcal{B}_{x}|,
\end{equation*}

where the subset convolution is over the max-sum semiring.
\end {observation}

By Proposition \ref{B2}, we compute $\mathcal{A}_{x}^{i}[g_{S},f,\widehat{c}]$ for every $S \subseteq \mathcal{B}_{x}\setminus R$ in  $2^{|\mathcal{B}_{x}\setminus R|}\cdot \tw^{\mathcal{O}(1)}\cdot \mathcal{O}(k \log k \log \log k)=2^{|\mathcal{B}_{x}\setminus R|}\cdot \tw^{\mathcal{O}(1)}$ time. Therefore, for each $i\in [k_{2}]$, we can compute $\mathcal{A}_{x}^{i}[g_{S},f,\widehat{c}]$ for every $S \subseteq \mathcal{B}_{x}\setminus R$ (or for every $d\in \mathcal{F}_{R}$) in $2^{|\mathcal{B}_{x}\setminus R|}\cdot 3^{|\widehat{c}|} \cdot \tw^{\mathcal{O}(1)}$ time. In the same time, we can compute $\mathcal{A}_{x}[d,f,\widehat{c}]$ by (\ref{eq5}) for every $d\in \mathcal{F}_{R}$. Also, we have to try all possible fixed subsets $R$ of $\mathcal{B}_{x}$. Since $\displaystyle\sum_{R\subseteq \mathcal{B}_{x}}2^{|\mathcal{B}_{x}\setminus R|}\leq 3^{|\mathcal{B}_{x}|}\leq 3^{(\tw+1)}$, the total time spent for all subsets $R\subseteq \mathcal{B}_{x}$ is $3^{\tw} \cdot \tw^{\mathcal{O}(1)}\cdot n$. Since for every $d$ there are at most $c^{\tw}$ choices for compatible $f$, so, clearly, the evaluation of all join nodes can be done in $(3c)^{\tw} \cdot \tw^{\mathcal{O}(1)}\cdot n$ time.

 Thus from the description of all nodes, we have the following theorem.
\cdisconnectedmatching*
\section{Lower Bound for Disconnected Matching} \label{DM}

In this section, we prove that assuming the Exponential Time Hypothesis, there does not exist any algorithm for the \textsc{Disconnected Matching} problem running in time $2^{o(\pw \log \pw)}\cdot n^{\mathcal{O}(1)}$. For this purpose, we give a reduction from \textsc{$k \times k$ Hitting Set}, which is defined as follows:
\bigskip 

\noindent\fbox{ \parbox{160mm}{
		\noindent \underline{\textsc{$k \times k$ Hitting Set:}}\\
		\textbf{Input:} A family of sets $S_{1},S_{2},\ldots, S_{m}\subseteq [k] \times [k]$ such that each set contains at most one element from each row of $[k]\times [k]$.\\
  \textbf{Parameter:} $k$.\\
		\textbf{Question:} Does there exist a set $\widehat{S}$ containing exactly one element from each row such that $\widehat{S}\cap S_{i}\neq \emptyset$ for every $i\in [m]$?}}
		
\begin{proposition}	[\cite{loksh}] \label{dmthm} Assuming Exponential Time Hypothesis, there is no $2^{o(k \log k)}\cdot n^{\mathcal{O}(1)}$-time algorithm for \textsc{$k \times k$ Hitting Set}.
\end{proposition}

Note that our reduction is inspired by the reduction given by Cygan et al. in \cite{cygan1} to prove that there is no $2^{o(\pw \log \pw)}\cdot n^{\mathcal{O}(1)}$-time algorithm for \textsc{Maximally Disconnected Dominating Set} unless the Exponential Time Hypothesis fails. Given an instance $(k,S_{1},\ldots, S_{m})$ of \textsc{$k \times k$ Hitting Set}, we construct an equivalent instance $(G,3k+m,k)$ of \textsc{Disconnected Matching} in polynomial time. First, we define a simple gadget that will be used in our construction.
\begin{definition}[Star Gadget] By adding a \emph{star gadget} to a vertex set $X\subseteq V(G)$, we mean the following construction: We
introduce a new vertex of degree $|X|$ and connect it to all vertices in $X$.

\end{definition}

\begin{lemma}\label{dm1}
Let $H$ be a graph and let $G$ be the graph constructed from $H$ by adding a star gadget to a subset $X$ of $V(H)$. Assume we are given a path decomposition $\widetilde{\mathcal{P}}$ of $H$ of width $\pw$ with the following property: There exists a bag in $\widetilde{\mathcal{P}}$ that contains $X$. Then, in polynomial
time, we can construct a path decomposition of $G$ of width at most $\pw +1$.
\end{lemma}
\begin{proof}
Let $u$ be the vertex introduced by the star gadget in $G$. By the assumption of the
lemma, there exists a bag, say, $\mathcal{B}_{u}$ in the path decomposition $\widetilde{\mathcal{P}}$ of $H$ that contains $X$. We introduce a new bag $\mathcal{B}_{u'}=\mathcal{B}_{u}\cup \{u\}$, and insert $\mathcal{B}_{u'}$ into $\widetilde{\mathcal{P}}$ after the bag $\mathcal{B}_{u}$. Let us call this new decomposition $\mathcal{P'}$. Next, we claim that $\mathcal{P'}$ is a path decomposition of $G$. As $N(u)=X$, $\mathcal{B}_{u'}$
covers all the edges incident
on $u$, and the rest of the edges are already taken care of by the bags of $\widetilde{\mathcal{P}}$ (as all the bags present in $\widetilde{ \mathcal{P}}$ are also present in $\mathcal{P'}$), this proves our claim. Moreover, note that we increased the maximum size of bags by at most one. Thus, the width of the path decomposition
$\mathcal{P'}$ is at most $\pw +1$.   
\end{proof}

If we attach a star gadget to multiple vertex disjoint subsets of $H$, then by applying Lemma \ref{dm1} (multiple times), we have the following corollary.
\begin{corollary}\label{corrdm}
Let $H$ be a graph and let $G$ be the graph constructed from $H$ by adding star gadgets to vertex disjoint subsets $X_{1},\ldots, X_{l}$ of $V(H)$. Assume we are given a path decomposition $\widetilde{\mathcal{P}}$ of $H$ of width $\pw$ with the following property: For each $X_{i}$, $i\in [l]$, there exists a bag in $\widetilde{\mathcal{P}}$ that contains $X_{i}$. Then, in polynomial
time, we can construct a path decomposition of $G$ of width at most $\pw +1$.
\end{corollary}

Now, consider the following construction.
\subsection {Construction}\label{dmc1}

Let $P_{i}=\{i\}\times [k]$ be a set containing all elements in the $i$-th row in the set $[k]\times[k]$. We define $\mathcal{S}=\{S_{s}:s\in [m]\} \cup \{P_{i}:i \in [k]\}$. Note that for each $X\in \mathcal{S}$, we have $|X|\leq k$, as each $S_{s}$, $s\in [m]$ contains at most one element
from each row and $|P_{i}|=k$ for each $i\in[k]$.

First, let us define a graph $H$. We start by introducing vertices $v_{i}^{L}$
for each $i\in [k]$ and vertices $v_{j}^{R}$
for each $j\in [k]$.
Then, for each set $X\in \mathcal{S}$, we introduce vertices $v_{i,j}^{X}$ for every $(i,j)\in X$. Let $V^{X}=\{v_{i,j}^{X}:(i,j)\in X\}$. We also introduce the edge set $\{v_{i}^{L}v_{i,j}^{X}\}\cup \{v_{i,j}^{X}v_{j}^{R}\}$ for each $X\in \mathcal{S}$ and $i,j\in [k]$. This ends the construction of $H$.

Now, we construct a graph $G$ from the graph $H$ as follows: For each $i\in [k]$ and $j\in [k]$, we attach
star gadgets to vertices $v_{i}^{L}$
and $v_{j}^{R}$. Furthermore, for each $X\in \mathcal{S}$, we attach star gadgets to $X$. For each $i\in [k]$ (resp. $j\in [k]$), let $u_{i}^{L}$ (resp. $u_{j}^{R}$) denote the unique vertex in the star gadget corresponding to $v_{i}^{L}$  (resp. $v_{j}^{R}$).  For each $X\in \mathcal{S}$, let $u^{X}$ denote the unique vertex in the star gadget corresponding to $X$. Let $E^{X}=\{v_{i,j}^{X}u^{X}:(i,j)\in X\}$. See Figure \ref{fig2} for an illustration of the construction of $G$ from $H$.

\begin{figure}[t]
    \centering
    \includegraphics[scale=0.68]{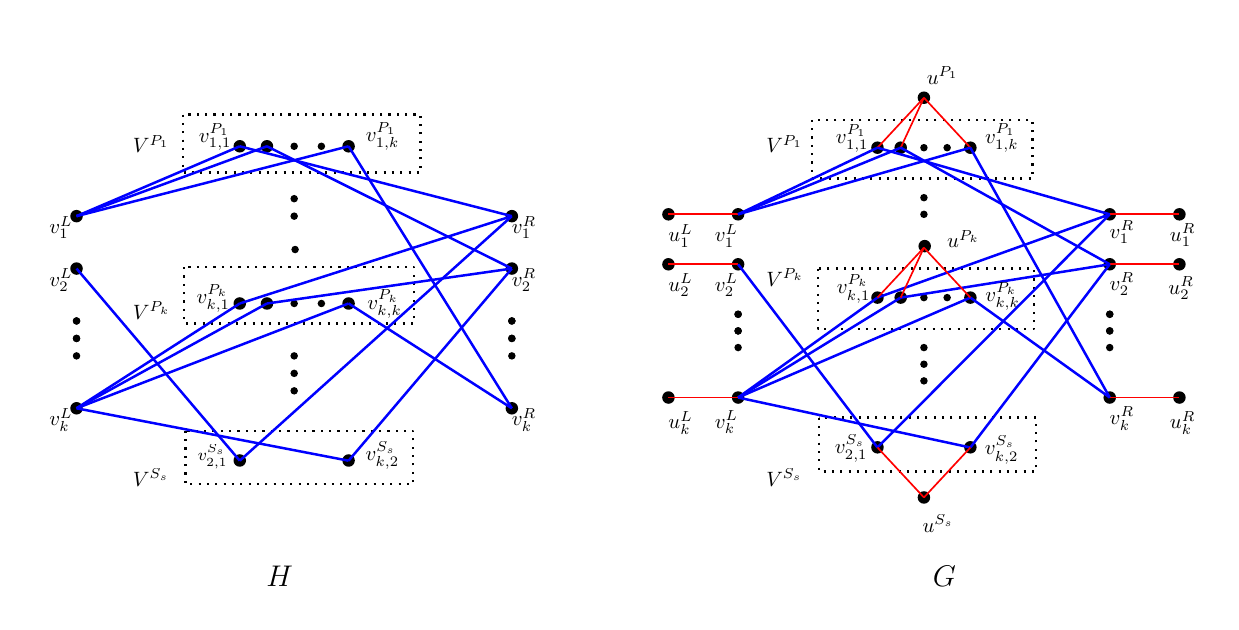}
    \caption{An illustration of the construction of $G$ from $H$. In this example, we assume that $S_{s}=\{(2,1),(k,2)\}.$ }
    \label{fig2}
\end{figure}

We now provide a pathwidth bound on $G$.

\begin{lemma}\label{dm2} Let $H$ and $G$ be as defined in Construction \ref{dmc1}. Then, the pathwidth of $G$ is at most $3k$.
\end{lemma}
\begin{proof}
First, consider the following path decomposition of $H$. For each $X\in \mathcal{S}$, we create a bag
 $$\mathcal{B}_{X}= V^{X} \cup \{v_{i}^{L}:i\in [k]\}\cup \{v_{j}^{R}:j\in [k]\}.$$

The path decomposition $\mathcal{\widetilde{P}}$ of $H$ consists of all bags $\mathcal{B}_{X}$ for $X\in \mathcal{S}$ in an arbitrary order. As $|X|\leq k$ for each $X\in \mathcal{S}$, the width of $\mathcal{\widetilde{P}}$ is at most $3k-1$. Note that in $\mathcal{\widetilde{P}}$, for every subset $B$ of $V(H)$, where a star gadget is attached, there exists a bag containing $B$. Therefore, by Corollary \ref{corrdm}, the width of $G$ is at most $3k$. Hence, the pathwidth of $G$ is at most $3k$.  
\end{proof}
\subsection {From Hitting Set to Disconnected Matching}\label{hitting}
\begin{lemma}\label{dm3} Let $G$ be as defined in Construction \ref{dmc1}. If the initial \textsc{$k \times k$ Hitting Set} instance is a Yes-instance, then there exists a matching $M$ in $G$ such that $|M|=3k+m$ and $G[V_{M}]$ has exactly $k$ connected components.
\end{lemma}
\begin{proof}
Let $\widehat{S}$ be a solution to the initial \textsc{$k \times k$ Hitting Set} problem instance $(k,S_{1},\ldots, S_{m})$. For each $X\in \mathcal{S}$, fix an element
$(i_{X},j_{X})\in \widehat{S}\cap X$. Since $\widehat{S}$ is a solution, it is clear that $\widehat{S}\cap S_{s} \neq \emptyset$ for each $s\in[m]$. Also, as $\widehat{S}$ contains exactly one element from each row, $\widehat{S}\cap P_{i} \neq \emptyset$ for each $i\in [k]$. Let
us define a matching $M$ in $G$ as follows:
$$ M=\{u_{i}^{L}v_{i}^{L}:i\in [k]\}\cup \{u_{j}^{R}v_{j}^{R}:j\in [k]\}\cup \{v_{i_{X},j_{X}}^{X}u^{X}:X\in \mathcal{S}\}.$$

First, note that $|M|=3k+m$, as there are $k$ vertices $v_{i}^{L}$ ($i\in [k]$), $k$ vertices $v_{j}^{R}$ ($j\in [k]$), and $|\mathcal{S}|=k+m$ (as $\mathcal{S}$ consists of $m$ sets $S_s$ ($s\in [m]$) and $k$ sets $P_{i}$ ($i\in [k]$)). As the endpoints of all the edges in $M$ are distinct, $M$ is a matching in $G$. Now, it remains to show that $G[V_{M}]$ contains exactly $k$ connected components. 
For this purpose, for each $j\in [k]$, let us define 

$$C_{j}=\{u_{j}^{R}v_{j}^{R}\}\cup \{u_{i}^{L}v_{i}^{L}:(i,j)\in \widehat{S}\}\cup \{v_{i_{X},j_{X}}^{X}u^{X}:X\in \mathcal{S}, j_{X}=j\}.$$

For some fixed $j\in [k]$, if $(i,j) \notin \widehat{S}$ for any $i\in [k]$, then $C_{j}=\{u_{j}^{R}v_{j}^{R}\}$, and hence connected. Otherwise, as $\{v_{j}^{R}v_{i_{X},j_{X}}^{X}:X\in \mathcal{S}, j_{X}=j\} \cup \{v_{i}^{L}v_{i_{X},j_{X}}^{X}: (i,j)\in \widehat{S}, X\in \mathcal{S}, j_{X}=j\}\subseteq E(G)$, $G[C_{j}]$ is a connected graph. Therefore, $G[C_{j}]$ is connected for each $j\in [k]$. Also, since $\widehat{S}$
contains exactly one element from each row, $G[C_{j}]$ and $G[C_{j'}]$ are disjoint for $j\neq j'$. Thus, $G[V_{M}]$ has exactly $k$ connected components.  
\end{proof}
\subsection {From Disconnected Matching to Hitting Set}\label{disc}
Let $G$ be as defined in Construction \ref{dmc1}. First, we partition the edges of $G$ into the following three types (see Figure \ref{fig3} for an illustration):
\begin{itemize}
    \item \textcolor{red}{Type-I} $=\{v_{i}^{L}u_{i}^{L}:i\in [k]\}\cup \{v_{j}^{R}u_{j}^{R}:j\in [k]\}$.
    \item \textcolor{teal}{Type-II} $=\{v_{i,j}^{X}u^{X}:X\in \mathcal{S}\}.$
    \item \textcolor{blue}{Type-III} $=\{v_{i}^{L}v_{i,j}^{X}: i\in [k], X\in \mathcal{S}\} \cup \{v_{j}^{R}v_{i,j}^{X}: j\in [k], X\in \mathcal{S}\}$.
\end{itemize}

\begin{figure}[t]
    \centering
    \includegraphics[scale=0.9]{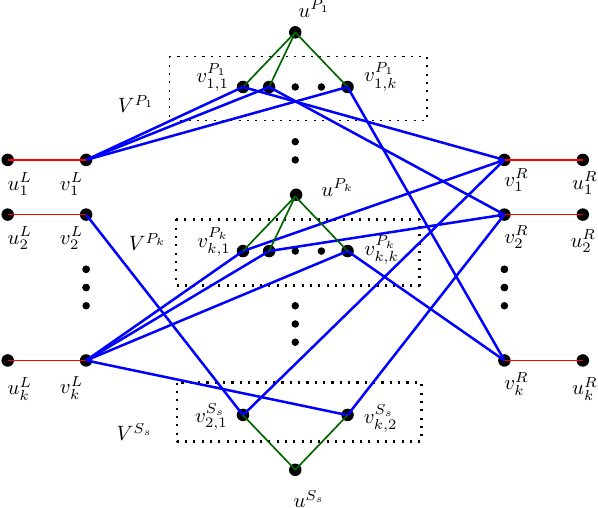}
    \caption{ Type-I, Type-II, and Type-III edges are represented by red, green, and blue colors, respectively. }
    \label{fig3}
\end{figure}
\medskip

From the definition of matching and the definitions of Type-I, Type-II, and Type-III edges, we have the following observation.
\begin {observation}\label{dmobs1}
Let $G$ be as defined in Construction \ref{dmc1}. For any matching $M$ in $G$, at most $2k$ edges from Type-I, at most $k+m$ edges from Type-II, and at most $2k$ edges from Type-III belong to $M$.  Furthermore, at most $2k$ edges from Type-I and Type-III combined belong to $M$.
\end {observation}
\begin{lemma} \label{dm4} Let $G$ be as defined in Construction \ref{dmc1}. If there exists a matching $M$ in $G$ such that $|M|\geq 3k+m$ and $G[V_{M}]$ has at least $k$ connected components, then there exists a matching $\widetilde{M}$ in $G$ such that $|\widetilde{M}|\geq 3k+m$, $G[V_{\widetilde{M}}]$ has at least $k$ connected components, and $\widetilde{M}$ contains edges from Type-I and Type-II only.
\end{lemma}
\begin{proof}
Let $M$ be a matching in $G$ such that $|M|\geq 3k+m$ and $G[V_{M}]$ has at least $k$ connected components. If $M$ contains edges from Type-I and Type-II only, then we are done. So assume that $M$ contains at least one edge from Type-III. Without loss of generality, let $v_{q}^{L}v_{q,j}^{X}\in M$ for some fixed (but arbitrary) $q,j\in [k]$ and $X\in \mathcal{S}$. Observe that $u_{q}^{L}$ is not saturated by $M$ as its only neighbor ($v_{q}^{L}$) is already saturated. Next, define $\widetilde{M}=(M\setminus v_{q}^{L}v_{q,j}^{X})\cup \{v_{q}^{L}u_{q}^{L}\}$. It is easy to see that $\widetilde{M}$ is a matching, and $|\widetilde{M}|=|M|$. Do this for every Type-III edge in $M$, that is, replace every Type-III edge in $M$ with its corresponding Type-I edge. As there are at most $2k$ Type-III edges in $G$, this can be done in polynomial time. By abuse of notation, let us call the matching so obtained $\widetilde{M}$. Now, it remains to prove that the number of connected components in $G[V_{\widetilde{M}}]$ is at least $k$. In each iteration, since we are replacing a non-pendant vertex with a pendant vertex, the number of connected components in $G[V_{\widetilde{M}}]$ cannot decrease in comparison with the number of connected components present in $G[V_{M}]$. Hence, the proof is complete.  
\end{proof}

\begin{lemma} \label{dm5} Let $G$ be as defined in Construction \ref{dmc1}. If there exists a matching $M$ in $G$ such that $|M|\geq 3k+m$ and $G[V_{M}]$ has at least $k$ connected components, then the initial \textsc{$k \times k$ Hitting Set} instance is a Yes-instance.
\end{lemma}
\begin{proof}
By Lemma \ref{dm4}, let $\widetilde{M}$ be a matching in $G$ such that $|\widetilde{M}|\geq 3k+m$, $G[V_{\widetilde{M}}]$ has at least $k$ connected components, and $\widetilde{M}$ contains edges from Type-I and Type-II only. First, we claim that $v_{i}^{L}$ and $v_{j}^{R}$ are saturated by $\widetilde{M}$ for each $i,j\in [k]$. To the contrary, without loss of generality, let $v_{q}^{L}$ is not saturated by $\widetilde{M}$ for some $q\in[k]$. Then, by Observation \ref{dmobs1}, at most $2k-1$ edges from Type-I and at most $k+m$ edges from Type-II can belong to $\widetilde{M}$. It implies that $|\widetilde{M}|\leq 3k+m-1$, a contradiction. The same arguments hold when we assume $v_{q}^{R}$ is not saturated by $\widetilde{M}$ for some $q\in[k]$.  Thus, $v_{i}^{L}$ and $v_{j}^{R}$ are saturated by $\widetilde{M}$ for each $i,j\in [k]$. 

Next, we claim that $|\widetilde{M}\cap E^{X}|=1$ for each $X\in \mathcal{S}$. Else, if for some $X'\in \mathcal{S}$, $\widetilde{M}\cap E^{X'}=\emptyset$, then by Observation \ref{dmobs1}, $|\widetilde{M}|\leq 3k+m-1$, a contradiction (to the fact that $|M|\geq 3k+m$).

Now, for each $i\in [k]$, let $v_{i,f(i)}^{P_{i}}$ be the unique vertex in $V_{\widetilde{M}}\cap V^{P_{i}}$. Let $\widehat{S}=\{(i,f(i)): i\in [k]\}$. We claim that
$\widehat{S}$ is a solution to the initial \textsc{$k \times k$ Hitting Set} instance. First, note that $\widehat{S}$ contains exactly one element from each row. Next, let $C_{j}$ be the connected component of $G[V_{\widetilde{M}}]$ that contains $v_{j}^{R}$. Note that $v_{i}^{L}\in C_{j}$ whenever
$j = f(i)$, i.e., $(i,j) \in \widehat{S}$. This implies that
$\bigcup_{j=1}^{k}C_{j}$ contains all vertices $v_{i}^{L}$. Moreover, as each vertex in $V^{X}$ for
$X\in \mathcal{S}$ is adjacent to some vertex $v_{j}^{R}$, $C_{j}'s$ are the only connected components of $G[V_{\widetilde{M}}]$. As $G[V_{\widetilde{M}}]$ contains at
least $k$ connected components, $C_{j}\neq C_{j'}$ for distinct $j,j' \in [k]$.

Next, let us consider those sets in $\mathcal{S}$, which correspond to the sets $S_{s}$, $s\in [m]$. Let $v_{i,j}^{S_{s}}$ be the unique vertex in $V_{\widetilde{M}}\cap V^{S_{s}}$. Note that $v_{i,j}^{S_{s}}$  connects $v_{i}^{L}\in C_{f(i)}$ with $v_{j}^{R}\in C_{j}$. As components $C_{j}'s$ are pairwise distinct, this implies that $j=f(i)$ and $(i,j)\in \widehat{S}\cap S_{s}$. Thus, the connected
components of $G[V_{\widetilde{M}}]$ are exactly the components $C_{j}$, $j\in [k]$.  
\end{proof}

By Proposition \ref{dmthm}, Lemma \ref{dm2}, Lemma \ref{dm3}, and Lemma \ref{dm5}, we have the following theorem.
\disconnectedmatching*

\section{Conclusions}
\label{sec:conclusions}
In this paper, we have studied \textsc{Induced Matching}, \textsc{Acyclic Matching}, $c$-\textsc{Disconnected Matching}, and \textsc{Disconnected Matching}, which are $\mathsf{NP}$-complete variants of the classical \textsc{Maximum Matching} problem, from the viewpoint of Parameterized
Complexity. We analyzed these problems with respect to the parameter treewidth, being, perhaps, the most well-studied parameter in the field.
\medskip

\noindent \textbf{$\mathsf{SETH}$-based Lower Bounds.} 
It would be of interest to prove or disprove whether the following is true:
\begin{conjecture}
	Unless the Strong Exponential Time Hypothesis (\textsf{SETH}) is false, there does not exist a constant $\epsilon>0$ and an
	algorithm that, given an instance $(G,\ell)$ together with a path decomposition of $G$ of width $\pw$, solves \textsc{Induced Matching} in $(3-\epsilon)^{\pw}\cdot n^{\mathcal{O}(1)}$ time.
\end{conjecture}

Given a graph $G$ and a positive integer $\ell$, in \textsc{Upper Dominating Set}, the aim is to find a minimal dominating set (i.e., a dominating set that is not a proper subset of any other dominating set) of cardinality at least $\ell$. Note that using the following proposition, one can achieve a lower bound for \textsc{Induced Matching} under the \textsf{SETH}.
\begin{proposition} [\cite{dublois}] \label{prop:uds}
Unless the \textsf{SETH} is false, there does not exist a constant $\epsilon>0$ and an algorithm that, given an instance $(G,\ell)$ together with a path decomposition of $G$ of width $\pw$, solves \textsc{Upper Dominating Set} in $(6-\epsilon)^{\pw}\cdot n^{\mathcal{O}(1)}$ time.
\end{proposition}

 Note that a dominating set is minimal if and only if it is also an irredundant set (a set of vertices $S$ in a graph such that for every $v\in S$, $
 N[S-\{v\}]\neq N[S])$. Given a graph $G$ and a positive integer $\ell$, \textsc{Irredundant Set} asks whether $G$ has an irredundant set of size at least $\ell$. To establish the $\mathsf{W}[1]$-hardness of \textsc{Induced Matching} (with respect to solution size as the parameter), Moser and Sikdar \cite{moser} gave a reduction from \textsc{Irredundant Set} to \textsc{Induced Matching} as follows: Given a graph $G$, where $V(G)=\{v_{1},\ldots,v_{n}\}$, construct a graph $H$ by making two copies $V'=\{v'_{1},\ldots,v'_{n}\}$ and $V''=\{v''_{1},\ldots,v''_{n}\}$ of $V(G)$ in $H$. Define $E(H) = \{ u'u'':u\in V(G)\} \cup \{u'v'', v'u'': uv\in E(G)\}$. Observe that if the pathwidth of $G$ is \pw, then the pathwidth of $H$ is at most $2\pw$. Also, note that $G$ has
an irredundant set of size $\ell$ if and only if $H$ has an induced matching of size $\ell$.  Therefore, by Proposition \ref{prop:uds}, we have the following theorem.

\begin{theorem} 
Unless the \textsf{SETH} is false, there does not exist a constant $\epsilon>0$ and an
	algorithm that, given an instance $(G,\ell)$ together with a path decomposition of $G$ of width $\pw$, solves \textsc{Induced Matching} in $(\sqrt{6}-\epsilon)^{\pw}\cdot n^{\mathcal{O}(1)}$ time.
\end{theorem}

Given a graph $G$, a vertex set $S \subseteq V(G)$ is an \textit{acyclic set} if $G[S]$ is an acyclic graph. Given a graph $G$ and a positive integer $\ell$, in \textsc{Maximum Induced Forest}, the aim is to find an acyclic set of cardinality at least $\ell$. Observe that \textsc{Maximum Induced Forest} is the complement of \textsc{Feedback Vertex Set}. Furthermore, it is known that unless the \textsf{SETH} is false, there does not exist a constant $\epsilon>0$ and an
	algorithm that, given an instance $(G,\ell)$ together with a path decomposition of $G$ of width \pw, solves \textsc{Feedback Vertex Set} in $(3-\epsilon)^\pw\cdot n^{\mathcal{O}(1)}$ time \cite{cygan}. Thus, we have the following corollary.
 \begin{corollary} 
Unless the \textsf{SETH} is false, there does not exist a constant $\epsilon>0$ and an
	algorithm that, given an instance $(G,\ell)$ together with a path decomposition of $G$ of width $\pw$, solves \textsc{Maximum Induced Forest} in $(3-\epsilon)^{\pw}\cdot n^{\mathcal{O}(1)}$ time.
\end{corollary}

Given an instance $(G,\ell)$ of \textsc{Maximum Induced Forest}, we construct an instance $(H,\ell)$ of \textsc{Acyclic Matching} by adding a pendant edge to every vertex of $G$. Observe that if the pathwidth of $G$ is \pw, then the pathwidth of $H$ is at most $\pw+1$. It is easy to see that $G$ has an acyclic set of size at least $\ell$ if and only if there exists an acyclic matching in $H$ of size at least $\ell$. Thus, we have the following theorem.

\begin{theorem} 
Unless the \textsf{SETH} is false, there does not exist a constant $\epsilon>0$ and an
	algorithm that, given an instance $(G,\ell)$ together with a path decomposition of $G$ of width $\pw$, solves \textsc{Acyclic Matching} in $(3-\epsilon)^{\pw}\cdot n^{\mathcal{O}(1)}$ time.
\end{theorem}

\noindent \textbf{Other Remarks.} Regarding \textsc{Acyclic Matching}, we note that the \emph{rank-based} method introduced by Bodlaender et al. \cite{bodl} can be used to derandomize our algorithm for \textsc{Acyclic Matching}, presented in Section \ref{AM}, in the standard way in which it is used to derandomize algorithms based on  Cut $\&$ Count. However, the dependence on the treewidth (specifically, the constant in the exponent) in the running time will become slightly worse. 

It is also noteworthy that the algorithms presented in Sections \ref{IM}-\ref{CDM} can be used to solve the optimization versions of their respective problems as well. Given a Yes-instance of \textsc{$\mathcal{P}$ Matching}, where $\mathcal{P}\in$ $\{\textsc{Induced}, \textsc{Acyclic}, c$-$\textsc{Disconnected}\}$, the idea is to remove an arbitrary vertex from the input graph and run the algorithm (presented in this paper) on the modified graph. If the modified graph becomes a No-instance, then the chosen vertex belongs to every solution of the input graph (here, a vertex belonging to a solution means that the respective matching saturates the vertex). Otherwise, we proceed with the modified graph and repeat the process. Each time, either we can identify a vertex that belongs to every solution, or we can reduce the size of the input graph by one vertex. Since we repeat the process at most $n$ times, where $n$ is the number of vertices in the input graph, our algorithm remains $\mathsf{FPT}$ (with a linear overhead in the running time). Future research directions could explore other structural parameterizations, such as vertex cover or feedback vertex set, with the aim of achieving faster running times or polynomial kernels. Additionally, there is room for improving the running time of the algorithms presented in this paper.

\begin{center}\textbf{Acknowledgments}\end{center}
 The authors are supported by the European Research Council (ERC) project titled PARAPATH.

\end{document}